\documentclass[12pt,draftclsnofoot,onecolumn]{IEEEtran}

\usepackage[numbers,sort&compress,square]{natbib}   % used for the advanced citer
\usepackage{graphicx}
\usepackage{subfigure}
\usepackage{amsmath,amssymb}
\usepackage{amsfonts}
\usepackage{array}
\usepackage{amsthm}
\usepackage{mathrsfs}
\usepackage{amsmath}
\usepackage{setspace}
\usepackage{color}
\usepackage[ruled,linesnumbered]{algorithm2e}
\usepackage{soul}

\newtheorem{theorem}{Theorem}
\newtheorem{proposition}[theorem]{Proposition}
\newtheorem{lemma}[theorem]{Lemma}

\theoremstyle{definition}
\newtheorem{define}[theorem]{Definition}

\newtheorem{problem}[theorem]{Problem}
\newtheorem{remark}[theorem]{Remark}
\newtheorem{example}[theorem]{Example}

\renewcommand{\baselinestretch}{1.4}

\begin{document}
	
	\title{Delay-Aware  Wireless Network Coding in Adversarial Traffic}

	\author{\IEEEauthorblockN{Yu-Pin Hsu}
		
		\footnote{Y.-P. Hsu is with  Department	of Communication Engineering, National Taipei University, Taiwan.  Email: \texttt{yupinhsu@mail.ntpu.edu.tw}.  An earlier version of Alg.~\ref{pda1} was presented at the Proceeding of IEEE International Symposium on Network Coding (NetCod) \cite{hsu2012opportunistic}. The work was supported by Ministry of Science and Technology of Taiwan under Grant MOST 107-2221-E-305-007-MY3. }
		
	}

	\maketitle
	
	\vspace{-2cm}
	\begin{abstract}
		We analyze a wireless line network employing wireless network coding.  The two end nodes exchange their packets through relays. While a packet at a relay  might not find its coding pair upon  arrival,  a transmission cost can be reduced by waiting for coding with a packet from the other side. To strike a balance between the  reduced transmission cost and the cost incurred by the delay, a scheduling algorithm  determining either to transmit an uncoded packet or  to wait for coding is needed.  Because of highly uncertain traffic injections, scheduling with no assumption of the traffic is critical.  This paper  proposes a \textit{randomized online scheduling algorithm} for a relay in  \textit{arbitrary} traffic, which can be \textit{non-stationary} or \textit{adversarial}.  The expected total cost (including a transmission cost and a delay cost) incurred by the proposed algorithm is at most $\frac{e}{e-1} \approx 1.58$ times the minimum achievable total cost. In particular, the proposed algorithm  is \textit{universal} in the sense that the ratio is independent of the traffic. With the universality,  the  proposed algorithm can be  implemented at each relay \textit{distributedly} (in a multi-relay network) with the same ratio.  Moreover, the proposed algorithm turns out to generalize the classic ski-rental online algorithm.

		%\begin{IEEEkeywords}
		%Wireless network coding, scheduling algorithms, adversarial traffic, competitive online algorithms.
		%\end{IEEEkeywords}
		
		%Our work uncover.....
		%{\color{red} more application: network coding, coding computing, index coding, 
		%
		%algorithm more flexible, work with adversarial channel and any MAC protocol
		%}

	\end{abstract}

	\section{Introduction}\label{section:introduction}
	\begin{figure}[t]
		\centering
		\includegraphics[width=.8\textwidth]{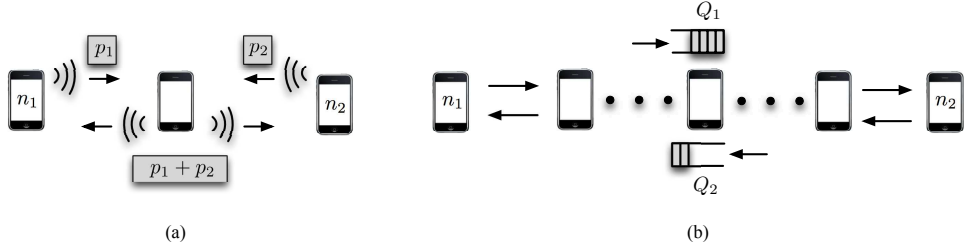}
		\caption{Relay networks with wireless network coding: (a) a single-relay network; (b) a multi-relay network.}
		\label{fig:network}
	\end{figure}
	
	There has been a  dramatic proliferation of research on \textit{wireless network coding}. The wireless network coding can substantially reduce the number of transmissions by exploiting the broadcast nature of wireless medium, resulting in power saving. Illustrated in Fig.~\ref{fig:network}-(a),  two  end nodes $n_1$ and $n_2$  exchange their respective packets $p_1$ and $p_2$ belonging to the Galois field $GF(2)$ through a relay. The conventional communication technique (without network coding) requires four transmissions (two for each packet). Leveraging the wireless network coding, only three transmissions are required; precisely,  nodes $n_1$ and $n_2$ send  packets $p_1$ and $p_2$ to  the relay, and  the relay  \textit{broadcasts} the scalar-linear combination $p_1+p_2$ (by bitwise XOR over $GF(2)$) to both end nodes. Each end node $n_i$ (for $i=1,2$) can recover its desired packet $p_{3-i}$ by subtracting (over $GF(2)$)  packet $p_i$ it already has from  packet $p_1+p_2$ it receives. In general, the wireless network coding can save up to 50\% of transmissions as long as the number of relays in a line network  in Fig.~\ref{fig:network}-(b) (as also called the \textit{reverse carpooling} \cite{effros2006tiling}) is large.

	To benefit from the wireless network coding, a  relay has to create sufficient coded packets; however,  a coded packet  at the relay can be created only when  packets from both sides are available. Precisely, a relay in Fig.~\ref{fig:network}-(b)  maintains two queues $Q_1$ and $Q_2$  storing packets from  both sides, respectively. If both queues are non-empty, then the relay can construct  coded packets by combining  packets from  both queues.  However, what should the relay do if only one queue is non-empty? Should the relay wait for  coding in the future or just transmit  uncoded packets from the non-empty queue? To fully realize the  advantage of the wireless network coding would incur  packet delays,  whereas  always transmitting  uncoded packets to minimize the delays  causes a larger number of transmissions. Therefore, a scheduling algorithm for determining when to code is crucial.

	The scheduling problem for a single-relay network as in Fig.~\ref{fig:network}-(a) under \textit{stationary stochastic} traffic has been investigated leveraging stochastic control techniques, like Lyapunov theory (e.g.,  \cite{ciftcioglu2011cost}) or Markov decision processes (e.g.,  \cite{hsu2014opportunities}). All the prior solutions fail to generalize to non-stationary or  adversarial (worst-case) traffic. In particular, they cannot be implemented  at each relay \textit{distributedly} in a multi-relay  network with provable performance guarantees. However, non-stationary or adversarial traffic has gained increasing importance in recent years. On one hand, external traffic injections at nodes $n_1$ or $n_2$ can arbitrarily be generated by their sources, following no particular probabilistic assumption.  On the other hand, the relay cannot expect the scheduling algorithms employed by nodes $n_1$ and $n_2$ to follow a stationary probabilistic distribution. In particular, \cite{borodin1996adversarial} claimed that the adversarial traffic is a better traffic model. 
	Because of those practical issues, the research on the adversarial traffic has attracted much attention in recent years (e.g., \cite{liang2018network}). Although  network coding design for adversarial channels has been an active area (e.g.,  \cite{ravagnani2018adversarial}), little attention was given to  the adversarial traffic in network-coding-enabled networks. To fill the gap, this paper aims to develop a \textit{universal} scheduling algorithm for arbitrary traffic with a provable performance guarantee.

	%
	%To handle the practical non-stationary or adversarial traffic in network-coding-enabled networks, we develop a  randomized online  scheduling algorithm for a relay, where the  algorithm does not have the arrival pattern as a prior. We show that the worst case ratio of the expected total costs, including a transmission cost and  a delay cost, between the proposed algorithm and an optimal scheduling algorithm is (asymptotically) at most $\frac{e}{e-1}$. As a result of the constant ratio regardless of the arrival pattern, the algorithm can be implemented at each relay in a multi-relay network, with guaranteeing the same ratio for each relay.
	Moreover, note that the  ski-rental problem  \cite{karlin2001dynamic} is a classic problem in an adversarial setting, where for each day a skier decides  either to buy a  ski or  to continue renting a ski  without knowing the skier's last vacation, e.g., the day when the snow melts. The ski-rental setting has been exploited in several works (e.g., \cite{lee2017online}) for  managing delays under some uncertainties. This paper shows that the proposed scheduling algorithm can  solve  generalized ski-rental scenarios.

	\subsection{Contributions}
	Our main contribution lies in designing and analyzing  scheduling for delay-aware wireless network coding in the adversarial traffic. The objective  is to minimize a total cost, including a transmission cost and a delay cost, for each relay. To reach the goal,  we show that our problem can be cast into a linear program. Leveraging \textit{primal-dual techniques} \cite{buchbinder2009design} for the linear program, we propose a \textit{randomized online scheduling algorithm} for each relay in a multi-relay network. In particular, the proposed algorithm can guarantee that the worst-case ratio between the expected total cost incurred by the proposed online  algorithm and that incurred by an optimal offline algorithm is (asymptotically)   $\frac{e}{e-1}\approx1.58$. 
	In addition to the theoretical worst-case analysis, the average-case analysis is conducted via computer simulations. Moreover, it turns out that the proposed algorithm can generalize the classic ski-rental algorithm to several scenarios.

	\subsection{Related works} 
	Scheduling design for network-coding-enabled networks has been  extensively explored from various perspectives. Most  scheduling works  with network coding aimed to maximize throughput (i.e., stability regions), e.g., \cite{parandehgheibi2010optimal,traskov2012scheduling,wiese2016scheduling,seferoglu2014network,moghadam2018lyapunov,ho2009dynamic,eryilmaz2011control,kuo2016robust,fragiadakis2014dynamic,jones2013distributed}. While \cite{parandehgheibi2010optimal,traskov2012scheduling,wiese2016scheduling,seferoglu2014network} considered static network environments and solved deterministic optimization problems, \cite{moghadam2018lyapunov,ho2009dynamic,eryilmaz2011control,kuo2016robust,fragiadakis2014dynamic,jones2013distributed} considered dynamic network environments and solved dynamic optimal control problems. Moreover, several scheduling works with network coding investigated delays, e.g., \cite{hou2014broadcasting,saif2018rate,swapna2013throughput}. In addition to the throughput or delays, some prior works analyzed other utilities or constraints when network coding is enabled, e.g., \cite{li2017joint} maximized a video reception quality and \cite{al2019cross} considered a Quality-of-Service (QoS) constraint.

	The most relevant works on  the trade-off between delays and  power consumption (with wireless network coding) in relay networks are \cite{ciftcioglu2011cost,hsu2014opportunities}. While \cite{ciftcioglu2011cost} proposed a scheduling algorithm using  Lyapunov techniques, \cite{hsu2014opportunities} showed the optimality of a threshold-type scheduling algorithm using Markov decision processes.  All those prior solutions were based on stochastic models with some stationary assumptions but cannot apply to non-stationary settings, especially in multi-relay networks. In contrast,  we  explore the trade-off in non-stationary settings.

	\section{System overview} \label{section:system}
	
	\subsection{Network model} \label{subsection:model}
	
	Consider a wireless line network in Fig.~\ref{fig:network}-(b). The two end nodes $n_1$ and $n_2$  send $N_1$ and $N_2$ packets (in $GF(2)$), respectively, to each other through  shared relay nodes. Divide  time into slots and index them by $t=1, 2, \cdots$. Suppose that a perfect schedule of wireless links  is given, so that  during each slot each node can transmit some packets under a transmission constraint without any interference. The interference-free link schedule can be achieved by existing medium access control (MAC) protocols, e.g.,  scheduled TDMA used in  \cite{hsu2014opportunities},  CSMA proposed by \cite{ni2011q}, or coded ALOHA proposed by \cite{paolini2015coded}. In fact, our design can work with any MAC protocol (see Remark~\ref{remark:any-mac} later).
	
	Consider a relay in Fig.~\ref{fig:network}-(b). The relay constructs queues for storing those packets that arrive at the relay but cannot be transmitted immediately upon arrival. As shown in Fig.~\ref{fig:network}-(b), the relay	maintains two queues $Q_1$ and $Q_2$ for packets  generated by nodes $n_1$ and $n_2$, respectively.  	At the beginning of each slot $t$, there are $A_1(t)$ new packets arriving at queue $Q_1$ and $A_2(t)$ new packets arriving at queue $Q_2$. By $\mathbf{A}=\{(A_1(1), A_2(1)), (A_1(2), A_2(2)), \cdots\}$ we define an \textit{arrival pattern} for the relay. The arrival pattern depends on the number of packets transmitted in the previous slot by its neighboring nodes. The arrival pattern is \textit{arbitrary}, which can be non-stationary or even adversarial. 	
	
	Let $Q_1(t)$ and $Q_2(t)$ be the number of packets at  queues $Q_1$ and $Q_2$, respectively, immediately after the packet arrivals in slot~$t$. If $Q_1(t)=0$ and $Q_{2}(t)=0$, then the relay  idles in slot~$t$. If $Q_{1}(t) \neq 0$ and $Q_{2}(t) \neq 0$, then the relay   transmits some coded  packets (under the transmission constraint) by combining (over $GF(2)$) packets from  both queues. Transmitting the coded packets can save the number of transmissions\footnote{To save the number of transmissions, both neighboring nodes of the relay must be able to decode the coded packets transmitted by the relay.  To that end, we leverage the reverse carpooling technique \cite{effros2006tiling}.   Each node (including both end nodes and all relay nodes) keeps  packets it transmitted previously for a while, so that when it receives a coded packet,  it can decoded the coded packet. See Footnote~\ref{footnote} later for the amount of time to keep a packet it previously transmitted. Moreover, each relay employs the decode-and-forward mechanism, where it decodes before re-encoding and transmitting  packets.} without incurring any delay. After transmitting the coded packets, if only one queue is non-empty and the relay can transmit more packets in that slot, then the relay has two  options for those packets at the non-empty queue: to \textit{transmit} some uncoded packets from  the non-empty queue or to \textit{idle} with the hope of receiving packets at the empty queue in the next slot (for coding). While always to transmit uncoded packets minimizes the  delays, always to idle minimizes the number of transmissions by  coding. To strike a balance between the delays and the number of transmissions,  the best decision is unclear when exactly one of the queues is non-empty.

		To investigate the best decision for each slot, we let $D(t)$ be the relay's \textit{decision} on the number of packets (including both uncoded and coded packets) transmitted in slot $t$.  We assume that  the broadcast channel from the relay to its neighboring nodes is noiseless. This simple model facilitates to explore the delays for coding in the arbitrary arrival pattern.  In fact, our design can  extend to adversarial ON-OFF channels (see Remark~\ref{remark:on-off-channel} later). Under the noiseless assumption, the queueing dynamics is 
	\begin{align*}
	Q_i(t+1)=\max\{Q_i(t)-D(t),0\}+A_i(t+1), 
	\end{align*}
	for all $i$ and $t$.  For example, if $Q_1(1)=5$, $Q_2(1)=3$, and $D(1)=4$, then the relay transmits three coded packets combining three packets from queues~$Q_1$ and $Q_2$ each, and transmits one uncoded packet from queue~$Q_1$; moreover, if $A_1(2)=2$, then $Q_1(2)=3$.

	A \textit{scheduling algorithm} $\pi=\{D(1), D(2), \cdots\}$ for the relay specifies decision $D(t)$ for each slot~$t$.   A scheduling algorithm is called an \textit{offline} scheduling algorithm if arrival pattern $\mathbf{A}$ is given as a prior. In contrast, a scheduling algorithm is called an \textit{online} scheduling algorithm if  arrival pattern~$\mathbf{A}$ (along with the numbers $N_1$ and $N_2$  of packets) is unavailable; instead, it knows the present arrivals $A_1(t)$ and $A_2(t)$ only, for each slot~$t$.

	\subsection{Problem formulation}
	To capture the trade-off between the delays and the number of transmissions, we  define a \textit{holding cost} and a \textit{transmission cost} as follows. Suppose that holding a packet at the end of a slot incurs a  cost of one unit. Moreover, suppose that each packet transmission takes a constant  cost of $C$ units\footnote{\label{footnote}The transmission cost $C$ is the weight (i.e., importance) between the transmission power of a packet and the delay of a packet for one slot, depending on  applications. If the transmission power is critical, then cost $C$ is larger; on the contrary, cost $C$ is smaller. Moreover, we are going to minimize the total transmission cost plus the total holding cost as in Eq.~(\ref{eq:total-cost}). In this context, the value of $C$ is the maximum number of slots for that a packet can delay. If a packet delays for more than $C$ slots, then it incurs more holding cost than the saving of transmission cost $C$ by coding. Thus, for a delay-sensitive application, we can set the value of $C$ to be its deadline constraint.}, where we assume that transmitting a coded packet incurs the same transmission cost as transmitting an uncoded packet. See Remark~\ref{remark:non-consistent-cost} for non-consistent costs for transmitting coded and uncoded packets. Moreover, we consider the case when the value of $C$ is greater than one\footnote{Following Footnote~\ref{footnote}, scheduling for the case when $C\leq 1$ is trivial: if only one queue is non-empty, then the relay always transmits uncoded packets but never waits for coding, because holding a packet for a slot incurs more cost than the saving of  transmission cost $C$ by coding. Thus,  this paper focuses on the case when $C> 1$, where packets can delay for some slots. }.

	Given arrival pattern $\mathbf{A}$,  we define a \textit{total cost} $J(\mathbf{A},\pi)$ under scheduling algorithm $\pi$  by
	\begin{align}
	J(\mathbf{A}, \pi)=&\sum^{\infty}_{t=1} C\cdot D(t) +\max\{Q_1(t)-D(t),0\}+\max\{Q_2(t)-D(t),0\},\label{eq:total-cost}
	\end{align}
	where the first term $C \cdot D(t)$ reflects the cost of transmitting $D(t)$ packets in slot~$t$ and the other terms $\max\{Q_1(t)-D(t),0\}+\max\{Q_2(t)-D(t),0\}$ reflects the cost of delaying all remaining packets for one slot. Since we consider the finite numbers $N_1$ and $N_2$, the minimum achievable total cost is finite. 
	
	We aim to develop an online scheduling algorithm such that the total cost is minimized for all possible arrival patterns $\mathbf{A}$. However, without knowing  arrival pattern~$\mathbf{A}$ (along with the total numbers $N_1$ and $N_2$ of packets) in advance,  an online scheduling algorithm is unlikely to achieve the minimum total cost (obtained by an optimal offline scheduling algorithm).   We characterize our online scheduling algorithm in terms of the \textit{competitiveness}  against an optimal offline scheduling algorithm, defined as follows. 
	\begin{define}	
		For arrival pattern $\mathbf{A}$, let $OPT(\mathbf{A}) = \min_{\pi} J(\mathbf{A},\pi)$ be the minimum total cost for all possible (offline) scheduling algorithms  $\pi$. Then, an online scheduling algorithm~$\pi$ is called \textbf{$\boldsymbol{\gamma}$-competitive} if 
		\begin{eqnarray*}
			J(\mathbf{A}, \pi) \leq \gamma \cdot OPT(\mathbf{A}),
		\end{eqnarray*}
		for all possible arrival patterns $\mathbf{A}$, where  $\gamma$ is called the \textbf{competitive ratio} of the online scheduling algorithm~$\pi$. 
	\end{define}
	
	\begin{remark}  \label{remark:multi-hop}
		A~$\gamma$-competitive online scheduling algorithm  guarantees that the resulting total cost is at most~$\gamma$ times the minimum total cost, \textit{regardless of  arrival patterns~$\mathbf{A}$}.  Thus, while a $\gamma$-competitive online scheduling algorithm can be  implemented at each relay  in the multi-relay network in a \textit{distributed} way, it guarantees the competitive ratio~$\gamma$ for each relay. 
	\end{remark}
	We aim to design and analyze an online scheduling algorithm for minimizing the competitive ratio.

	\section{One-sided adversarial traffic} \label{section:one-side}
	We start with a fixed number of packets waiting for coding; in particular, this section focuses on the following setting:
	\begin{enumerate}
		\item Queue $Q_{1}$ has all  $N_1$ (with $N_1 \leq C$) packets initially, i.e., $A_{1}(1)=N_1$ and  $A_{1}(t)=0$ for all $t \geq 2$.
		\item Queue $Q_2$ is injected by arbitrary traffic with a total of $N_2$ packets.
		\item The relay  can transmit any number of packets in each slot. 
	\end{enumerate}
	The setting is referred to as the \textit{one-sided adversarial traffic}.  With the  first and second assumptions,  we can focus on a fixed number $N_1$ of packets at queue $Q_1$ waiting for coding, while capturing the key feature of the adversarial arrival pattern at queue $Q_2$. In fact, the first assumption is practical as well for bursty  traffic at queue $Q_1$. Section~\ref{section:two-side} will also generalize to two-sided adversarial traffic.  Note that, under the first assumption,  the relay never delays the packets at queue $Q_2$ for minimizing the total cost. 
	%Moreover, without loss of generality, we can assume $N_1 \geq N_2$. 
	%If $N_1 < N_2$, then those $N_2-N_1$ more packets at queue $Q_2$ cannot affect an optimal decision for queue $Q_1$ and hence  can be removed from the scheduling design. 
	The third assumption is made for delivering a clear insight into our innovation.  Lemma~\ref{lemma:number-of-tx} will analyze the maximum number of transmissions required by the proposed online  scheduling algorithm; moreover, Section~\ref{subsection:constraint} will extend to a transmission constraint.

		\subsection{Overview of our methodology} \label{subsection:overview}
		This section provides an overview of our methodology (leveraging \textit{primal-dual techniques}  \cite{buchbinder2009design} for linear programs):
		\begin{enumerate}
			%	\item We propose an integer program in integer program~(\ref{integer}) for \textit{optimally} solve our scheduling problem in the \textit{offline} fashion (with arrival pattern~$\mathbf{A}$ as a prior). 
			%	\item By relaxing the integer variables in the integer program, we propose a linear program in linear program~(\ref{primal1}) for \textit{optimally} solving our scheduling problem in the \textit{offline} fashion. 
			\item We propose linear program~(\ref{primal1}) for \textit{optimally} solving our scheduling problem in the \textit{offline} fashion (with  arrival pattern~$\mathbf{A}$ as a prior).
			\item We propose Alg.~\ref{pda1} for \textit{sub-optimally} solving linear program~(\ref{primal1}) in the \textit{online} fashion  (without  arrival pattern~$\mathbf{A}$ as a prior).
			\item 	We analyze the objective value (of linear program~(\ref{primal1})) computed by Alg.~\ref{pda1} through a solution (produced also by Alg.~\ref{pda1}) to the dual of the linear program.  We show that the objective value computed by Alg.~\ref{pda1} is no more than $\frac{e}{e-1}$ times that dual objective value. Then, the duality theory yields that  the objective value computed by Alg.~\ref{pda1} is no more than $\frac{e}{e-1}$ times the minimum objective value (of linear program~(\ref{primal1})).
			\item By transforming the fractional solution produced by Alg.~\ref{pda1} to  randomized decisions, we propose a randomized online scheduling algorithm in Alg.~\ref{online-alg1}.
			\item We show that the expected cost incurred by Alg.~\ref{online-alg1} is no more than the objective value computed by Alg.~\ref{pda1}. Then, by the third bullet, the expected total cost incurred by  Alg.~\ref{online-alg1} is also no more than $\frac{e}{e-1}$ times the minimum objective value (i.e., minimum achievable total cost).
		\end{enumerate}
		
		Section~\ref{subsection:primal-dual} formulates the linear program (as a primal program) and its dual program. While Section~\ref{subsection:primal-dual-alg} proposes Alg.~\ref{pda1} for solving the primal program and the dual program in the online fashion, Section~\ref{subsection:analysis} analyzes the solution produced by Alg.~\ref{pda1}. Leveraging the solution produced by Alg.~\ref{pda1}, Section~\ref{subsection:online-alg1} proposes Alg.~\ref{online-alg1} for solving our scheduling problem and analyzes its expected total cost.

	\subsection{Primal-dual formulation}\label{subsection:primal-dual}
	Given  arrival pattern $\mathbf{A}$, this section casts the \textit{offline} scheduling problem (under the one-sided adversarial traffic) into a linear program.  To that end, we introduce  some variables: 
	\begin{itemize}
		\item $x$: the number of packets at queue $Q_1$  transmitted \textit{without coding}.
		\item $z(t)$: the number of packets at queue $Q_1$   at the end slot  of~$t$.
	\end{itemize}
	If the relay decides to transmit $x$ uncoded packets at queue~$Q_1$, then it must\footnote{If one of those $x$ uncoded packets is transmitted in slot $t>1$, then the total cost in Eq.~(\ref{eq:J-new}) increases by $t$ (for holding the packet for $t$ slots). For minimizing the total cost, the $x$ uncoded packets are optimally transmitted in slot~1. } transmit the $x$ uncoded packets  in slot~$1$. Thus, the total cost  in Eq.~(\ref{eq:total-cost}) under the one-sided adversarial traffic  can be expressed by
	\begin{align}	
	J(\mathbf{A},\pi)= C \cdot N_2+ C \cdot x +\sum_{t=1}^{\infty}  z(t), \label{eq:J-new}
	\end{align}
	where the first term $C \cdot N_2$  is the cost of transmitting  all packets at queue~$Q_2$, the second term  $C \cdot x$ is the cost of transmitting the $x$ uncoded packets at queue~$Q_1$ in slot~1 (i.e., transmitting coded packets at queue~$Q_1$ is free), and the last term $\sum_{t=1}^{\infty}  z(t)$ is the cost incurred by holding the $N_1-x$ packets at queue~$Q_1$.

	By removing the constant  $C \cdot N_2$ from Eq.~(\ref{eq:J-new}), we have the following scheduling problem.
	\begin{problem} \label{problem1}
		Under the one-sided adversarial traffic, develop a scheduling algorithm for the packets at queue~$Q_1$ such that the cost $C \cdot x +\sum_{t=1}^{\infty}  z(t)$ is minimized.
	\end{problem}

	\begin{remark}\label{remark:special}
		This remark shows that the classic ski-rental problem \cite{karlin2001dynamic} is a special case of our Problem~\ref{problem1}.
			In the  ski-rental problem, a skier arrives at a resort on day 1 with no ski. For each day, the skier decides either to  buy a ski or to rent a ski. If the skier buys a ski in a day, then the skier does not have to rent a ski after that day. While renting a ski for a day takes one dollar, buying a ski takes $C$ dollars. The skier will stay at the resort for $T$ days until the last vacation day. The goal is to minimize the  buying cost  plus the total renting cost. Given the instance of the ski-rental problem, we construct an instance of our Problem~\ref{problem1}. We construct one packet for each queue, i.e., $N_1=N_2=1$. We construct a packet staying at queue~$Q_1$ in slot~1 (corresponding to the skier). We construct a packet arriving at queue~$Q_2$ in slot~$T$ (corresponding to the last vacation day). Next, we link a skier's decision with a relay's decision.	While the skier rents a ski on day $t$ if and only if the relay idles in slot~$t$, the skier buys a ski on day $t$ if and only if  the relay transmits the packet at queue~$Q_1$ without coding in slot~$t$. While the skier does not have to make decisions after day~$T$, the relay also does not have to make decisions after slot~$T$ (because the relay can transmit a coded packet in slot~$T$ if the packet at $Q_1$ still stays at that queue in slot~$T$).  With the link between the ski-rental problem and our Problem~\ref{problem1}, variable $x$ in Problem~\ref{problem1} can indicate if the skier buys a ski,  and variable $z(t)$ in Problem~\ref{problem1} can indicate if the skier rents a ski on  day~$t$.  Suppose that  holding the packet at queue~$Q_1$ for a slot takes one dollar, and that transmitting an uncoded packet from queue~$Q_1$ takes $C$ dollars. Then, the value of $C \cdot x +\sum_{t=1}^{\infty}  z(t)$ in Problem~\ref{problem1} can represent the buying cost plus the total renting cost.  Thus, the ski-rental problem equivalently becomes our Problem~\ref{problem1}. In other words, the ski-rental problem is a special case ($N_1=N_2=1$) of our Problem~\ref{problem1}.
	\end{remark}

	\begin{remark}\label{remark:link1}
	Following Remark~\ref{remark:special}, this remark shows that our Problem~\ref{problem1} is a generalization of the ski-rental problem.  We can think of each packet at queue~$Q_1$ as a skier and think of a slot when a packet arrives at queue~$Q_2$ as the day when a skier has to leave. Moreover, buying a ski takes $C$ dollars while renting a ski for a day takes one dollar.  With the transformation, Problem~\ref{problem1} considers a \textit{group} of  skiers (i.e., the $N_1$ packets at queue $Q_1$) with potentially different last vacation days (i.e., the arriving slots at queue $Q_2$).  Those skiers cooperatively make a buying or renting decision on each day for minimizing the total  buying cost  plus the total renting cost. 
\end{remark}

	\begin{remark} \label{remark:non-consistent-cost}
		If transmitting a coded packet incurs a cost of $C_1$ units and transmitting an uncoded packet incurs a different cost of $C_2$ units with $C_1 > C_2$, then the total cost in Eq.~(\ref{eq:total-cost}) becomes
		\begin{align*}
		J(\mathbf{A},\pi)= C_2 \cdot N_2+ C_2 \cdot x +(C_1-C_2)(N_1-x)+\sum_{t=1}^{\infty}  z(t), 
		\end{align*}
		where the term $(C_1-C_2) (N_1-x)$ is the extra cost for transmitting the coded packets.  Then, we can   replace cost $C$ in Problem~\ref{problem1} with $2 C_2 - C_1$. Note that $2 C_2 - C_1 \geq 0$. If $C_1$ were higher than $2 C_2$, then transmitting a coded packet by combining two packets would not save any cost from transmitting two uncoded packets. The rest of the paper  focuses on the constant cost $C$ without loss of generality.  
	\end{remark}

	Next, we propose the following integer program for optimally solving Problem~\ref{problem1} in the offline fashion:
	
	\textbf{Integer program:}	
	\begin{subequations} \label{integer}
		\begin{eqnarray}	
		&\min & C \cdot x +\sum_{t=1}^{\infty}  z(t) \label{integer:opjective}\\
		&\text{s.t.}& x+z(t) \geq N_1-n_2(t) \,\,\,\text{for all $t$}; \label{integer:constraint-1}\\
		&& x, z(t) \in \mathbb{N}\,\,\, \text{for all $t$},  \label{integer:constraint-2}
		\end{eqnarray}
	\end{subequations}
	where $n_2(t)=\sum^{t}_{\tau=1} A_2(\tau)$ is  the total number of packets arriving at queue~$Q_2$ until slot $t$.  The constraint in Eq.~(\ref{integer:constraint-1}) is because for each slot $t$ the number of packets at queue $Q_1$  is at least $N_1-x-n_2(t)$, where $x$ packets at queue $Q_1$ are transmitted  without coding in slot $t=1$ and at most $n_2(t)$ packets at queue $Q_1$ are transmitted with coding by slot $t$.

	Next, by  relaxing the integrality constraint in Eq.~(\ref{integer:constraint-2}) to real numbers, we  obtain the following linear program.
	
	\textbf{Linear program (primal program):}	
	\begin{subequations} \label{primal1}
		\begin{eqnarray}	
		&\min & C \cdot x +\sum_{t=1}^{\infty}  z(t) \label{primal1:objective}\\
		&\text{s.t.}& x+z(t) \geq N_1-n_2(t) \,\,\,\text{for all $t$}; \label{primal1:constraint-1}\\
		&& x, z(t)\geq 0\,\,\, \text{for all $t$}.  \label{primal1:constraint-2}
		\end{eqnarray}
	\end{subequations}
	After the relaxation,  a feasible \textit{fractional} solution for~$x$ in linear program~(\ref{primal1}) can no longer represent a decision for the number of packets at queue~$Q_1$ transmitted without coding (but an \textit{integral} solution for $x$ in linear program~(\ref{primal1}) can). In fact, the next lemma shows that the relaxation has no integrality gap.
	\begin{lemma} \label{lemma:integrality}
		The relaxation from integer program~(\ref{integer}) to linear program~(\ref{primal1}) has no integrality gap. 
	\end{lemma}
	\begin{proof}
		Suppose that an optimal solution to linear program~(\ref{primal1}) is non-integral. Then, we establish a contradiction in Appendix~\ref{appendix:lemma:integrality}. 
	\end{proof}
	From Lemma~\ref{lemma:integrality}, Problem~\ref{problem1}  can be optimally solved in polynomial time if arrival pattern $\mathbf{A}$ is given in advance: Solve for variable $x$ in linear program~(\ref{primal1})； then, transmit $x$ uncoded packets in slot~1. After transmitting the uncoded packets in slot~$1$, all other packets at queue~$Q_1$ always wait for packets at queue $Q_2$  for coding. 
		
		Next, while Section~\ref{subsection:primal-dual} proposes an online algorithm for sub-optimally solving for variable $x$ without knowing arrival pattern $\mathbf{A}$ in advance, Section~\ref{subsection:analysis} analyzes the objective value in Eq.~(\ref{primal1:objective}) computed by the proposed online algorithm by its dual program. Thus, we refer to  linear program~(\ref{primal1}) as a primal program and express its  dual program  as follows.

	\textbf{Dual program:}	
	\begin{subequations} \label{dual1}
		\begin{eqnarray}
		&\max& \sum^{\infty}_{t=1} (N_1-n_2(t))  w(t)  \label{dual1:objective}\\
		&\text{s.t.}& \sum_{t=1}^{\infty} w(t) \leq C;  \label{dual1:constraint1}\\
		&& 0 \leq w(t) \leq 1 \,\,\text{for all $t$}. \label{dual1:constraint2}
		\end{eqnarray}
	\end{subequations}

	\subsection{Primal-dual  algorithm}\label{subsection:primal-dual-alg}
	This section proposes a \textit{primal-dual  algorithm} in Alg.~\ref{pda1} for obtaining a  solution to  primal program~(\ref{primal1}) and dual program~(\ref{dual1}). The primal-dual  algorithm does not have  arrival pattern $\mathbf{A}$ as a prior; instead,  it can obtain the present arrivals $A_1(t)$ and $A_2(t)$ only, for each slot~$t$.  
%	{\color{red}Note that the solution to dual program~(\ref{dual1}) produced by the primal-dual  algorithm is just used to analyze the primal objective value in Eq.~(\ref{primal1:objective}) computed by the  algorithm.}
%	

	\begin{algorithm}[t]
			\small
		\SetAlgoLined 
		\SetKwFunction{Union}{Union}\SetKwFunction{FindCompress}{FindCompress} \SetKwInOut{Input}{input}\SetKwInOut{Output}{output}
		
		%\Input{Bids vector $B$ and side information $H_i$ for all clients}
		%\Output{Encoding matrix $\Gamma$}
		
		\tcc{Initialize all variables at the beginning of slot $1$ as follows:}
		
		$x$, $z(t)$, $w(t)$ $\leftarrow 0$ for all $t$\; \label{pda1:initial}
		$x_1 , \cdots, x_{N_1}, z_1(t), \cdots, z_{N_1}(t) \leftarrow 0$ for all $t$\tcp*[r]{Auxiliary variables.} \label{pda1:more-variable}
		$\theta \leftarrow (1+\frac{1}{C})^{\lfloor C \rfloor}-1$\tcp*[r]{$\theta$ is a constant with the function of cost~$C$}. 	\label{pda1:constant}

		\tcc{For each new slot $t=1, 2, \cdots$, the variables are updated as follows:}
		
		%\SetVline
		%\If{$n_2(t) < N_1$}{
		\For{$i=n_2(t)+1$ \emph{\KwTo} $N_1$\label{pda1:for}}{   
			\If{$x_i < 1$ \label{pda1:condition}}{
				$z_i(t) \leftarrow 1-x_i$\; \label{pda1:zi}
				$x_i \leftarrow x_i(1+\frac{1}{C})+\frac{1}{\theta \cdot C}$\; \label{pda1:xi}
				$w(t) \leftarrow 1$\;\label{pda1:y}
				
			}
		}
		$z(t) \leftarrow \sum_{i=1}^{N_1} z_i(t)$\; \label{pda1:z}	
		$x \leftarrow \sum^{N_1}_{i=1} x_i$\; \label{pda1:x}
		%}	
		\caption{Primal-dual   algorithm for solving primal program~(\ref{primal1}) and dual program~(\ref{dual1}).}
		\label{pda1}
	\end{algorithm}
	
	Alg.~\ref{pda1} initializes all variables (in Lines \ref{pda1:initial} and \ref{pda1:more-variable}) at the beginning of slot 1. Obtaining the present arrivals $A_1(t)$ and $A_2(t)$ at the beginning of each new slot~$t$, Alg.~\ref{pda1} updates all variables for slot $t$.  For updating the value of $x$, Alg.~\ref{pda1} introduces a set of auxiliary variables $x_1, \cdots x_{N_1}$ (initialized in Line~\ref{pda1:more-variable}). The intuition\footnote{The intuition here is just our idea of solving for  variable~$x$ in the online fashion, but is not  scheduling decisions for packets. Section~\ref{subsection:online-alg1} will cast a value of variable $x$ to a randomized decision.} behind updating variable $x_i$ and $x$ in Lines~\ref{pda1:condition},~\ref{pda1:xi}, and~\ref{pda1:x} is following:	We can imagine the value of $x_i$ to be a probability of transmitting the $i$-th (counted from the head of queue $Q_1$) packet at queue $Q_1$ without coding. Precisely, for each  slot~$t$, Line~\ref{pda1:xi}  increases the value of $x_i$  for those packets potentially staying at queue $Q_1$:
	\begin{itemize}
		\item A total of $n_2(t)$ packets arrive at queue $Q_2$ by slot~$t$, yielding at most $n_2(t)$ coded packets until $t$. As such, Line~\ref{pda1:for} considers  $x_i$, for $i=n_2(t)+1$ until $N_1$, because only the $(n_2(t)+1)$-th packet until the $N_1$-th packet might wait at queue $Q_1$ in slot $t$, but other packets have been transmitted with coding by slot $t$.
		\item Moreover, if the value of $x_i$ is greater than or equal to one (i.e., the condition in Line~\ref{pda1:condition} fails), then the \mbox{$i$-th} packet has been transmitted without coding by slot $t$. Thus, Line~\ref{pda1:xi} updates only those $x_i$'s satisfying the condition in Line~\ref{pda1:condition}.
	\end{itemize}
	The constant $\theta$ used in Line~\ref{pda1:xi} is specified as the function of transmission cost $C$  in Line~\ref{pda1:constant} for satisfying the dual constraint in Eq.~(\ref{dual1:constraint1}).  Then, Line~\ref{pda1:x} sets the value of $x$ to be that of $\sum^{N_1}_{i=1} x_i$, counting all $N_1$ packets at queue~$Q_1$.

	Moreover, Alg.~\ref{pda1} introduces another set of auxiliary variables  $z_1(t), \cdots, z_{N_1}(t)$ for all $t$, and updates the value of $z_i(t)$ to be that of $1-x_i$ (in Line~\ref{pda1:zi})  in slot $t$ for satisfying the constraint in Eq.~(\ref{primal1:constraint-1}). Again, the value of $z(t)$ is set to be that of $\sum_{i=1}^{N_1} z_i(t)$ in Line~\ref{pda1:z},  counting all  $N_1$ packets at queue $Q_1$. In addition, the value of $w(t)$ is updated to be one in Line~\ref{pda1:y} for maximizing the dual objective value in Eq.~(\ref{dual1:objective}) subject to the constraints in Eq.~(\ref{dual1:constraint2}).

	We want to emphasize that the solution produced by Alg.~\ref{pda1} can be non-integral. The solution is just a feasible solution to  primal program~(\ref{primal1}) but can no longer represent the number of packets at queue~$Q_1$ transmitted without coding. However, by exploiting the  solution produced by Alg.~\ref{pda1},  Section~\ref{subsection:online-alg1} will  propose a randomized online  scheduling algorithm for solving Problem~\ref{problem1}. The underlying idea is that the intermediate fractional solution for $x$ in primal program~(\ref{primal1})  produced by Alg.~\ref{pda1} in each slot can be transformed to a probability of transmitting an uncoded packet in that slot.

	\subsection{Analysis of Alg.~\ref{pda1}} \label{subsection:analysis}
	This section analyzes the proposed Alg.~\ref{pda1}. Since the values of all variables can be updated by Alg.~\ref{pda1} in each slot,  the following proofs  use $\widehat{x}(t)$, $\widehat{z}(t)$, $\widehat{w}(t)$, $\widehat{x}_i(t)$, $\cdots$, $\widehat{x}_{N_1}(t)$, $\widehat{z}_1(t)$, $\cdots$, $\widehat{z}_{N_1}(t)$ to represent the corresponding values at the \textit{beginning} (before update) of slot~$t$; use $\widetilde{x}(t)$, $\widetilde{z}(t)$, $\widetilde{w}(t)$, $\widetilde{x}_i(t)$, $\cdots$, $\widetilde{x}_{N_1}(t)$, $\widetilde{z}_1(t)$, $\cdots$, $\widetilde{z}_{N_1}(t)$ to represent the corresponding values at the \textit{end} (after update) of slot $t$.  Note that  $\widetilde{x}(\infty)$, $\widetilde{z}(t)$, $\widetilde{w}(t)$, $\widetilde{x}_i(\infty)$, $\cdots$, $\widetilde{x}_{N_1}(\infty)$, $\widetilde{z}_1(t)$, $\cdots$, $\widetilde{z}_{N_1}(t)$ is the solution produced by Alg.~\ref{pda1}. In fact, Alg.~\ref{pda1} will not update any variable after slot $\lfloor C \rfloor$ (see Remark~\ref{remark:c} later), where we recall that the value of $C$ is the transmission cost.

	The next lemma establishes  the primal feasibility of Alg.~\ref{pda1}.
	\begin{lemma} \label{lemma:feasible-primal1}
		Alg.~\ref{pda1} produces a feasible solution to  primal program~(\ref{primal1}).
	\end{lemma}
	\begin{proof}
		See Appendix~\ref{appendix:lemma:feasible-primal1}.
	\end{proof}
	
	%For verifying the dual feasibility of Alg.~\ref{pda1}, we need the following technical lemma. 
	%\begin{lemma} \label{lemma:x-ineq}
	%For each slot $t$, if $\widetilde{w}(t)=1$, then we have 
	%\begin{align}
	%\widetilde{x}_i(t) \geq \frac{(1+\frac{1}{C})^t-1}{\theta}. \label{eq:x-ineq}
	%\end{align}
	%for all $i=\max(n_2(t)+1, N_1), \cdots, N_1$. 
	%\end{lemma}
	%\begin{proof}
	%We prove the claim 	in Eq.~(\ref{eq:x-ineq})  by induction on slot $t$. First, the claim is true for $t=1$ since $\widetilde{x}_i(1) \geq \frac{1}{\theta \cdot C}$ (see Line~\ref{pda1:xi}) if $\widetilde{w}(1)=1$. Second, suppose that the claim is true for slot $t-1$. Next, we consider slot $t$ in the following.  Suppose that in slot $t$, the value of $w(t)$ is set to be one (i.e., $\widehat{x}_i(t)<1$ for  $i=\max(n_2(t)+1, N_1), \cdots, N_1$). Then, the claim in Eq.~(\ref{eq:x-ineq}) is true according to Line~\ref{pda1:xi} again:
	%\begin{align*}
	%\widetilde{x}_i(t) =& \widetilde{x}_i(t-1)(1+\frac{1}{C})+\frac{1}{\theta \cdot C}\\
	%\mathop{\geq}^{(a)}& \frac{(1+\frac{1}{C})^{t-1}-1}{\theta}(1+\frac{1}{C})+\frac{1}{\theta \cdot C}\\
	%=&\frac{(1+\frac{1}{C})^t-1}{\theta},
	%\end{align*}
	%where (a) is due to the induction hypothesis for $t-1$. 
	%\end{proof}
	
	The next lemma establishes the dual feasibility of Alg.~\ref{pda1}. For proving the lemma, we define the increment (under Alg.~\ref{pda1}) of the value of $x_i$  in slot $t$ by $\Delta \widetilde{x}_i(t)=\widetilde{x}_i(t)-\widehat{x}_i(t)$. From \cite{buchbinder2009design}, we can obtain that  $\Delta \widetilde{x}_i(1)=\frac{1}{\theta \cdot C}$ for all $i$; moreover, if $\widetilde{w}(t)=1$, then $\Delta \widetilde{x}_i(t)=(1+\frac{1}{C}) \Delta \widetilde{x}_i(t-1)$ for  $i=n_2(t)+1, \cdots N_1$. That is, $\{\Delta \widetilde{x}_i(1), \Delta \widetilde{x}_i(2), \cdots\}$ forms a geometric sequence, with the initial value  of $\frac{1}{\theta \cdot C}$ and the ratio of $1+\frac{1}{C}$.
	
	\begin{lemma} \label{lemma:dual}
		Alg. \ref{pda1} produces a feasible  solution to  dual program~(\ref{dual1}).
	\end{lemma}
	
	\begin{proof}
		First, the dual constraint in Eq.~(\ref{dual1:constraint2}) holds obviously according to Lines~\ref{pda1:initial} and \ref{pda1:y}.  Second, we will show that the value of $\sum_{t=1}^{\infty} w(t)$ in Eq.~(\ref{dual1:constraint1}) computed by Alg.~\ref{pda1} is less than or equal to $C$. Note that Line~\ref{pda1:y} updates the value of $w(t)$ to be one  if the condition in Line~\ref{pda1:condition} holds. Thus, it suffices to show that the condition in Line~\ref{pda1:condition} fails at the end of slot  $\lfloor C \rfloor$, i.e., the value of $w(t)$  is zero for all $t > \lfloor C \rfloor$. 
		
		If $\widetilde{w}(\lfloor C \rfloor)=1$ in slot $\lfloor C \rfloor$, then we can obtain 
		\begin{align*}
		\widetilde{x}_i(\lfloor C \rfloor)= \sum_{t=1}^{\lfloor C \rfloor} \Delta \widetilde{x}_i(t)\mathop{=}^{(a)} \frac{(1+\frac{1}{C})^{\lfloor C \rfloor} -1}{\theta}=1,
		\end{align*}
		where (a) is because the sequence $\{\Delta\widetilde{x}_i(1),\Delta \widetilde{x}_i(2), \cdots, \Delta \widetilde{x}_i(\lfloor C \rfloor)\}$ is the geometric sequence with the initial value of $\frac{1}{\theta \cdot C}$ and the ratio of $1+\frac{1}{C}$.
		Thus, the value of $w(t)$ for all $t > \lfloor C \rfloor$ is zero since the condition in Line~\ref{pda1:condition} fails. 
	\end{proof}
	
	\begin{remark}\label{remark:c}
		According to the proof of Lemma~\ref{lemma:dual}, Alg.~\ref{pda1}  no longer updates the values of all variables after slot $\lfloor C \rfloor$. 	
	\end{remark}
	
	The next theorem analyzes the primal objective value in Eq.~(\ref{primal1:objective}) computed by  Alg.~\ref{pda1}.

	\begin{theorem} \label{theroem:competitive-ratio1}
		Let $OPT_{(\ref{primal1})}(\mathbf{A})$ be the minimum objective value in linear program~(\ref{primal1}). Then, 
		the primal objective value in Eq.~(\ref{primal1:objective}) computed by Alg.~\ref{pda1} is bounded above by
		\begin{align*}
		(1+\frac{1}{(1+\frac{1}{C})^{\lfloor C \rfloor}-1}) OPT_{(\ref{primal1})}(\mathbf{A}),
		\end{align*}
		for all possible arrival patterns $\mathbf{A}$. 
	\end{theorem}
	\begin{proof}
		Let $\Delta \mathscr{P}(t)$ be the increment (under Alg.~\ref{pda1})  of the primal objective value in Eq.~(\ref{primal1:objective}) in slot $t$ and let $\Delta \mathscr{D}(t)$ be that of the dual objective value in Eq.~(\ref{dual1:objective}) in slot $t$.  Appendix~\ref{appendix:theroem:competitive-ratio1} establishes that 
		\begin{align*}
		\Delta \mathscr{P}(t) \leq \left(1+\frac{1}{\theta} \right)\Delta \mathscr{D}(t), 
		\end{align*}
		for all $t$. Let $\mathscr{P}$ and $\mathscr{D}$ be the primal and dual objective values, respectively, computed by Alg.~\ref{pda1}. Then, $\mathscr{P}=\sum_{t=1}^{\infty} \Delta \mathscr{P}(t)$ and $\mathscr{D}=\sum_{t=1}^{\infty} \Delta \mathscr{D}(t)$; therefore, the result follows since
		\begin{align*}
		\mathscr{P}\leq \left(1+\frac{1}{\theta} \right) \mathscr{D} \leq \left(1+\frac{1}{\theta} \right) OPT_{(\ref{primal1})}(\mathbf{A}),
		\end{align*}
		where the last inequality is due to the weak duality \cite{buchbinder2009design}.
	\end{proof}

	\subsection{Randomized online scheduling algorithm}\label{subsection:online-alg1}
	\begin{algorithm}[t]
			\small
		\SetAlgoLined 
		\SetKwFunction{Union}{Union}\SetKwFunction{FindCompress}{FindCompress} \SetKwInOut{Input}{input}\SetKwInOut{Output}{output}
		
		\tcc{Initialize all variables at the beginning of slot $1$ as follows:}
		
		$x, x_1 , \cdots, x_{N_1} \leftarrow 0$\; \label{online-alg1:initial}
		
		$\theta \leftarrow (1+\frac{1}{C})^{\lfloor C \rfloor}-1$\tcp*[r]{$\theta$ is a constant with the function of cost~$C$}. 	\label{online-alg1:constant}
		
		Pick a uniformly random number $u \in [0,1)$\; \label{online-alg1:u}
		
		\tcc{For each new slot $t=1,2,  \cdots$, perform as follows:}

		Transmit $\min(Q_1(t), A_2(t))$ coded packets\; \label{online-alg1:code}
		
		\tcc{After transmitting the coded packets, if queue~$Q_1$ is non-empty, then continute as follows:}

		\For{$i=n_2(t)+1$ \emph{\KwTo} $N_1$\label{online-alg1:for}}{   
			\If{$x_i < 1$ \label{online-alg1:condition}}{
				$x_i \leftarrow x_i(1+\frac{1}{C})+\frac{1}{\theta \cdot C}$\; \label{online-alg1:xi}
			}
		}\label{online-alg1:for-end}
		
		$x_{\text{pre}} \leftarrow x$\; \label{online-alg1:pre}
		$x \leftarrow \sum^{N_1}_{i=1} x_i$\; \label{online-alg1:x}
		
		\While{1\label{online-alg1:while}}{
			\uIf{$x_{\text{pre}} \leq u < x$ \label{online-alg1:if}}{
				Transmit an uncoded packet from queue~$Q_1$\; \label{online-alg:tx-uncoded}
				$u \leftarrow u+1$\;  \label{online-alg1:u+1}
			} \label{online-alg1:if-end}
			\Else(\tcp*[f]{$x \leq u$}){\label{online-alg1:else-break}
				break\; \label{online-alg1:break}
			} \label{online-alg1:break-end}
		}\label{online-alg1:while-end}

		\caption{Randomized  online scheduling algorithm for  the one-sided adversarial traffic model.}
		\label{online-alg1}
	\end{algorithm}
	Leveraging Alg.~\ref{pda1}, this section proposes a randomized online scheduling algorithm in Alg.~\ref{online-alg1}.  For each slot~$t$, Alg.~\ref{online-alg1} transmits $\min(Q_1(t), A_2(t))$ coded packets (in Line~\ref{online-alg1:code}) by combing packets left at queue~$Q_1$ and the new arriving packets at queue~$Q_2$. Then, to decide whether to transmit  uncoded packets for each slot, Lines~\ref{online-alg1:for} - \ref{online-alg1:for-end} and \ref{online-alg1:x} update the values of $x_i$ and $x$  in the same way as Alg.~\ref{pda1} does. 
	In addition, Alg.~\ref{online-alg1} uses another variable $x_{\text{pre}}$ (in Line~\ref{online-alg1:pre}) to record the value of $x$ at the beginning (before update in Line~\ref{online-alg1:x}) of each slot. Let $\widetilde{x}_{\text{pre}}(t)$ be the value of $x_{\text{pre}}$ at the end of slot $t$.

	At the beginning of slot $1$, Line~\ref{online-alg1:u} chooses a random number $u \in [0,1)$ from a continuous uniform distribution between 0 and 1. According to Lines~\ref{online-alg1:while} - \ref{online-alg1:while-end}, if there exists a $k \in \mathbb{N}$ such that $u+k \in [\widetilde{x}_{\text{pre}}(t), \widetilde{x}(t))$, then the relay transmits an uncoded packet in slot $t$. Note that, if there are multiple~$k$'s such that $u+k \in [\widetilde{x}_{\text{pre}}(t), \widetilde{x}(t))$, then the relay  transmits multiple uncoded packets in slot $t$, until the present value of $u$ is greater than or equal to $\widetilde{x}(t)$ (as in Line~\ref{online-alg1:break}). 
	
	Let $\Delta \widetilde{x}(t)=\widetilde{x}(t)-\widehat{x}(t)$ ($=\widetilde{x}(t)-\widetilde{x}_{\text{pre}}(t)$) be the increment of the value of $x$  in slot~$t$. The idea behind Alg.~\ref{online-alg1} is that, with the random choice of $u$, the expected number of uncoded packets transmitted in slot~$t$ is exactly $\Delta \widetilde{x}(t)$.

	\begin{theorem} \label{theorem:expected-competitive ratio1}
		The expected competitive ratio of  Alg.~\ref{online-alg1} is
		\begin{align*}
		1+\frac{1}{(1+\frac{1}{C})^{\lfloor C \rfloor}-1},
		\end{align*}
		approaching  $\frac{e}{e-1}$ as $C$ tends to infinity.
	\end{theorem}
	\begin{proof}
		We show that   the expected  cost of transmitting uncoded packets by Alg.~\ref{online-alg1} is  \mbox{$C \cdot \sum_{t=1}^{\infty}\Delta \widetilde{x}(t)=C\cdot \widetilde{x}(\infty)$}, which is the value of the first term in Eq.~(\ref{primal1:objective}) computed by Alg.~\ref{pda1}. Moreover, we show that the expected number of packets left at queue $Q_1$ at the end of slot $t$ under Alg.~\ref{online-alg1} is less than or equal to $\widetilde{z}_i(t)$, which is the value of the second term in Eq.~(\ref{primal1:objective}) computed by Alg.~\ref{pda1}. Thus, the expected cost incurred by Alg.~\ref{online-alg1} is less than or equal to the primal objective value in Eq.~(\ref{primal1:objective}) computed by Alg.~\ref{pda1}. Then, the result immediately follows from Theorem~\ref{theroem:competitive-ratio1}. See Appendix~\ref{appendix:theorem:expected-competitive ratio1} for details. 
	\end{proof}
	
	\begin{remark}
	Recall that a competitive ratio is the worst-case ratio for all possible cases (i.e., arrival patterns~$\mathbf{A}$) and recall that the  ski-rental problem is a  case  of our Problem~\ref{problem1} (from Remark~\ref{remark:special}). Thus, the minimum achievable competitive ratio for our problem is no higher than  that for the ski-rental problem. Because the minimum achievable competitive ratio for the ski-rental problem  is  \mbox{$\frac{e}{e-1}$} \cite{karlin2001dynamic} and Alg.~\ref{online-alg1} can also  achieve that competitive ratio, we can conclude that Alg.~\ref{online-alg1}   achieves the minimum achievable competitive for Problem~\ref{problem1}. 
	\end{remark}
	
	\begin{remark}\label{remark:any-mac}
		We want to emphasize that the competitive ratio in Theorem~\ref{theorem:expected-competitive ratio1} is independent of  arrival patterns $\mathbf{A}$,i.e., regardless of the MAC protocol. Thus, Alg.~\ref{online-alg1} can  be implemented at each relay in  the multiple-relay network; meanwhile, it can ensure the same competitiveness for each relay. 
	\end{remark}

	The next lemma investigates the maximum number of uncoded packets per slot required  by  Alg.~\ref{online-alg1}.  
	
	\begin{lemma} \label{lemma:number-of-tx}
		Alg.~\ref{online-alg1} transmits at most three uncoded packets in each slot.
	\end{lemma}
	\begin{proof}
		Since $\{\Delta \widetilde{x}_i(1), \Delta \widetilde{x}_i(2), \cdots, \Delta \widetilde{x}_i(C)\}$ is the geometric sequence with the initial value of $\frac{1}{\theta \cdot C}$ and the ratio of $1+\frac{1}{C}$ for all $i$,  we have  
		\begin{eqnarray*}
			\Delta \widetilde{x}(t)=\sum_{i=1}^{N_1} \Delta \widetilde{x}_i(t) \leq \frac{N_1}{\theta \cdot C} (1+ \frac{1}{C})^{C-1}. 
		\end{eqnarray*}
		Moreover, because of $(1+\frac{1}{C})^{C-1} \leq 3$, $\theta \geq 1$, and $N_1 \leq C$ (from the assumption for the one-sided traffic), we have $\Delta \widetilde{x}(t) \leq 3$. Thus, at most three $k$'s such that $u+k \in [\widetilde{x}_{\text{pre}}(t), \widetilde{x}(t))$, i.e., Alg.~\ref{online-alg1} transmits at most three uncoded packets in each slot. 
	\end{proof}

	%		To analyze the computational complexity of Alg.~\ref{online-alg1}, we recall from Remark~\ref{remark:c} that all variable $x_i$'s have the values of more than one after $\lfloor C \rfloor$ slots, i.e., there are at most $\lfloor C \rfloor$ iterations during Lines~\ref{online-alg1:for} - \ref{online-alg1:for-end}. Thus, the complexity of Alg.~\ref{online-alg1} is $O(C)$. As the value of $C$ grows, the complexity of  Alg.~\ref{online-alg1} increases but the competitive ratio in Theorem~\ref{theorem:expected-competitive ratio1} decreases. 

	To analyze the computational complexity of Alg.~\ref{online-alg1}, we note that there are at most $N_1$ iterations in Lines~\ref{online-alg1:for} - \ref{online-alg1:for-end}. Moreover, there are at most $3$ iterations  in Lines~\ref{online-alg1:while} - \ref{online-alg1:while-end} (by Lemma~\ref{lemma:number-of-tx}). Since $N_1\leq C$ (from the assumption for the one-sided traffic), the computational complexity of Alg.~\ref{online-alg1} is $O(C)$. As the value of $C$ grows, the computational complexity  increases but the competitive ratio in Theorem~\ref{theorem:expected-competitive ratio1} decreases.

\section{Two-sided adversarial traffic} \label{section:two-side} 
This section relaxes the first assumption in the one-sided adversarial traffic by allowing arbitrary traffic at both queues $Q_1$ and $Q_2$. We start with the scenario where only packets at a queue can wait for coding; in particular, this section starts with the following setting:
		\begin{enumerate}
			\item The packets at queue $Q_1$ can wait for coding  but those at queue $Q_2$ are transmitted immediately upon arrival.
			\item The relay can transmit any number of packets in each slot. 
		\end{enumerate}
		This setting is referred to as the \textit{two-sided adversarial traffic}.   This model can  make us focus on  decisions for a queue  while capturing the key feature of the two-sided adversarial traffic. In fact, the first assumption is practical as well when the traffic generated by~node $n_2$ is urgent and even cannot delay for more than one slot (e.g., urgent events in intelligent transportation systems or ultra-reliable low-latency communications (URLLC) \cite{azari2019risk} in 5G). Later,  Section~\ref{subsection:both-queue} will relax the first assumption by extending to the general case when packets at both queues can wait for coding. Moreover, Section~\ref{subsection:constraint} will relax the second assumption by imposing a transmission constraint.

	We introduce some variables similar to Section \ref{section:one-side}:
	\begin{itemize}
		\item $x_{i}$: indicate if  the $i$-th packet  at queue $Q_1$ is transmitted \textit{without coding} upon arrival, where $x_{i}=1$ if the  packet is transmitted without coding; $x_{i}=0$ otherwise;
		%	\item $y_i$: 	 indicate if  the $i$-th packet  at queue $Q_2$ is transmitted \textit{without coding}, where $y_{i}=1$ if the  packet is transmitted without coding and $y_{i}=0$ otherwise;	
		\item  $z(t)$: the number of packets at queue~$Q_1$ at the end of slot $t$.
	\end{itemize}
	%$n_2(t)=\sum^{t}_{\tau=1} A_2(\tau)$:  the total number of packets arriving at $Q_2$ until slot $t$.
	%Without loss of generality, we can assume that $n_1(t) \geq n_2(t)$ for all $t$. We have the following scheduling problem.
	We have the following problem similar to Problem~\ref{problem1}.
	\begin{problem}\label{problem2}
		Under the two-sided adversarial traffic, develop a scheduling algorithm for the packets at queue~$Q_1$ such that the cost $C \cdot \sum_{i=1}^{N_1}x_i + \sum_{t=1}^{\infty}z(t)$ is minimized. 
	\end{problem}
	
	\begin{remark}
		Following the argument in Remark~\ref{remark:link1}, Problem~\ref{problem2} considers a group of  skiers  \textit{arriving arbitrarily} with potentially different last vacation days. Those skiers cooperatively make a buying or renting decision in each day for minimizing the total buying cost  plus the total renting cost.
	\end{remark}

	%Then,  the scheduling problem can be solved by the following integer program, similar to integer program~(\ref{integer}).
	%
	%\textbf{Integer program:}
	%\begin{subequations} \label{interger2}
	%	\begin{align}	
	%	\min& \,\,\,\,C \cdot \sum^{N_1}_{i=1} x_i + \sum_{t=1}^{\infty} z(t) \label{integer2:objective}\\
	%	\text{s.t.}&\,\,\,\,\sum^{n_1(t)}_{i=1} x_i +z(t)= n_1(t)-n_2(t)+\sum^{n_2(t)}_{j=1} y_j \,\,\,\text{for all $t$} \label{integer2:constraint-1}\\
	%		&\,\,\,\, x_i, y_j, z(t) \in \mathbb{N}\,\,\, \text{for all $i$, $j$, and $t$}.  \label{integer2:constraint-2}
	%	\end{align}
	%\end{subequations}
	%The constraint in Eq.~(\ref{integer2:constraint-1}) is because all  packets at queue $Q_1$ arriving by slot $t$ (a total of $n_1(t)$ packets) can be: (1) transmitted without coding (a total of $\sum_{i=1}^{n_1(t)}x_i$ packets), (2) combined with packets at queue $Q_2$ (a total of $n_2(t)-\sum^{n_2(t)}_{j=1} y_j$ packets), or (3) waiting at slot $t$ (a total of $z(t)$ packets). 
	
	%Because of the similarity to integer program~(\ref{integer}),  Alg.~\ref{pda1} is a candidate for solving integer program~(\ref{interger2}). However, a key issue arises as shown in the next example.

	Section~\ref{subsection:idea} discusses  ideas underlying another primal-dual formulation that will be proposed by Section~\ref{subsection:primal-dual-2} for solving Problem~\ref{problem2}. With the new primal-dual formulation, Section~\ref{subsection:prima-dual-alg2} proposes a primal-dual algorithm for solving Problem~\ref{problem2} in the online fashion.

	\subsection{Ideas underlying the primal-dual formulation} \label{subsection:idea}

	The next example shows that an immediate extension from linear program~(\ref{primal1}) along with Alg.~\ref{pda1} cannot solve Problem~\ref{problem2} with the competitive ratio in Theorem~\ref{theroem:competitive-ratio1}.
	
	\begin{example} \label{ex:wrong}
		Suppose that two packets arrive at queue $Q_1$ in slots 1 and 3, respectively, and no packet arrives at queue~$Q_2$.  Assume  transmission cost $C=2$. 
		In this case,  the optimal solution to Problem~\ref{problem2} is $x_{1}=1$ and $x_{2}=1$, i.e., both packets at queue $Q_1$ are optimally transmitted without coding upon arrival. In particular,
		the optimal solution satisfies the following linear program (similar to linear program~(\ref{primal1})). 
		
		\textbf{Linear program (primal  program):}	
		\begin{subequations} \label{primal-worng}
			\begin{eqnarray}	
			&\min& 2(x_{1}+x_{2})+ \sum_{t=1}^{\infty} z(t) \label{primal-worng:objective}\\
			&\text{s.t.}& x_{1} +z(t) \geq  1 \,\,\,\text{for $t=1,2$;} \label{primal-wrong:constraint1}\\
			&& x_{1}+x_{2}+z(t) \geq 2 \,\,\, \text{for $t=3,4, \cdots$}; \label{primal-wrong:constraint2}\\
			&&x_1, x_2, z(t) \geq 0 \,\,\,\text{for all $t$}.
			\end{eqnarray}
		\end{subequations}
		The associated dual program can be expressed as
		
		\textbf{Dual program:}
		\begin{subequations}\label{dual:wrong}
			\begin{eqnarray}	
			&\max& \sum^{2}_{t=1} w(t) + 2 \cdot \sum_{t=3}^{\infty} w(t) \\
			&\text{s.t.}& \sum^{\infty}_{t=1} w(t) \leq 2; \label{dual-wrong:constraint}\\
			&&0\leq w(t) \leq 1\,\,\,\text{for all $t$}.
			%	& & z_0 \geq \mathcal{K} \label{eq:primal-const2}\\
			%	&& z(t) - z_{t+1} \leq 1 \,\,\text{for all}\,\, t  \label{eq:primal-const3}\\
			%	&& x_i, x_{2,j}, z(t) \in \{0,1,  \cdots \}\,\,\, \text{for all}\,\, 1 \leq i \leq N_1, 1\leq j \leq N_2, t \geq 1. 
			\end{eqnarray}
		\end{subequations}	
		
		Applying the idea behind Alg.~\ref{pda1}, we would update  $x_i \leftarrow x_i(1+\frac{1}{C})+\frac{1}{\theta \cdot C}$ and update  $w(t) \leftarrow 1$ until the dual constraint in Eq.~(\ref{dual-wrong:constraint}) becomes tight. Given $C=2$, the constant $\theta$ is $(1+\frac{1}{2})^2-1=\frac{5}{4}$. In slot~1,  update  $x_1$ to be $\frac{1}{\frac{5}{4} \cdot 2}=\frac{2}{5}$ and update $w(1)$ to bo one. In slot~2, update $x_2$ to be $\frac{2}{5}(1+\frac{1}{2})+\frac{1}{\frac{5}{4} \cdot 2}=1$ and update $w(2)$ to be one. Because the dual constraint in Eq.~(\ref{dual-wrong:constraint}) becomes tight in slot~2, we cannot update any variable since slot~3; in particular, we cannot update $x_3$ when the second packet arrives at queue~$Q_1$. Thus, the second packet waits forever, yielding an infinite  holding cost. 
	\end{example}
	
	To tackle the issue in the above example, the next example proposes another primal-dual formulation. 
	
	%The example displays that the idea of Alg.~\ref{pda1} can solve linear program~(\ref{primal-worng}) if the constraint in Eq.~(\ref{primal-wrong:constraint2}) is substituted by a set of \textit{equivalent} constraints. 
	
	\begin{example} \label{ex:decouple}
		Note that an optimal solution for $x_1$ and $x_2$ in linear program~(\ref{primal-worng}) also satisfies the following linear program, where we use $z_{1}(t)$ and $z_{2}(t)$ to indicate if  the first packet and second packet, respectively, stay at queue $Q_1$ at the end of slot~$t$. 
		
		\textbf{Linear program (primal  program):}
		\begin{subequations} \label{primal-revise}
			\begin{eqnarray}	
			&\min& 2(x_{1}+x_{2})+ \sum_{t=1}^{\infty} z_{1}(t) + \sum_{t=1}^{\infty} z_{2}(t)\label{primal-revise:objective}\\
			&\text{s.t.}& x_{1} +z_1(t) \geq  1 \,\,\,\text{for $t=1, 2, \cdots$;} \label{primal-revise:constraint1}\\
			&& x_{2}+z_2(t) \geq 1 \,\,\, \text{for $t=3, 4, \cdots$}; \label{primal-revise:constraint2}\\
			&&x_1, x_2, z_1(t), z_2(t)\geq 0 \,\,\,\text{for all $t$}.
			\end{eqnarray}
		\end{subequations}
		While expressing  variable $z(t)$ in Eq.~(\ref{primal-worng:objective})  by $z_1(t)+z_2(t)$ in Eq.~(\ref{primal-revise:objective}), we substitute the original constraints in Eqs.~(\ref{primal-wrong:constraint1}) and~(\ref{primal-wrong:constraint2})  by the constraints in Eqs.~(\ref{primal-revise:constraint1}) and (\ref{primal-revise:constraint2}). The associated dual program can be expressed as
		
		\textbf{Dual program:}
		\begin{subequations} \label{dual-revise}
			\begin{eqnarray}	
			&\min& \sum^{\infty}_{t=1} w_{1}(t) + \sum_{t=3}^{\infty} w_{2}(t) \\
			&\text{s.t.}& \sum^{\infty}_{t=1} w_{1}(t) \leq 2; \\
			&& \sum_{t=3}^{\infty} w_{2}(t)  \leq 2;\label{dual-revise:constraint2}\\
			&& 0 \leq w_1(t), w_2(t) \leq 1\,\,\,\text{for all $t$}.
			\end{eqnarray}
		\end{subequations}
		Follow the idea behind Alg.~\ref{pda1} as discussed in Example~\ref{ex:wrong}. In slot~1,  update $x_1$ to be $\frac{2}{5}$ and update $w_1(1)$ to be one. In slot~2,  update $x_1$ to be one and update $w_1(2)$ to be one. In slot~3,  update $x_2$ to be $\frac{2}{5}$ and update $w_2(3)$ to be one. In slot~4,  update $x_2$ to be one and update $w_2(4)$ to be one. The updating process can achieve the  competitive ratio  in Theorem~\ref{theroem:competitive-ratio1}.
	\end{example}

	The above example implies that the idea of Alg.~\ref{pda1} can solve Problem~\ref{problem2} with  the same competitive ratio, if we can formulate a linear program with  constraints for each \textit{individual} packet (like Eqs.~(\ref{primal-revise:constraint1}) and~(\ref{primal-revise:constraint2})) instead of those for \textit{all} arriving packets (like Eqs.~(\ref{primal-wrong:constraint1}) and~(\ref{primal-wrong:constraint2})). 
	%if we can decouple the  original constraint in Eq.~(\ref{integer2:constraint-1}) into a set of equivalent constraints for each packets at queue~$Q_1$, then Alg.~\ref{pda1} can solve the two-sided traffic scenario 
	%In other words, we want to consider  packet-by-packet constraint  (like Eqs.~(\ref{primal-revise:constraint1}) and (\ref{primal-revise:constraint2})), instead of the total number $z(t)$ of packets (like Eqs.~(\ref{primal-wrong:constraint1}) and~(\ref{primal-wrong:constraint2})). 
	In this context, we introduce additional variables: let $z_{i}(t)$ indicate if  the $i$-th packet stays at queue $Q_1$ at the end slot $t$, where $z_{i}(t)=1$ if it does and   $z_{i}(t)=0$ otherwise. For each slot~$t$,  the value of $x_i+z_{i}(t)$ is either zero or one, where  $x_i+z_{i}(t)=0$ implies that  the $i$-th packet at  queue $Q_1$ is transmitted with coding by slot~$t$ and  $x_i+z_i(t)=1$ implies that the packet is either transmitted without coding  by slot~$t$ or stays at queue $Q_1$ at the end of slot~$t$. By the next example, we emphasize that  the constraints should be carefully considered.

	%Following the idea in Example~\ref{ex:decouple}, for each slot $t$ we can  decouple the original constraint in Eq.~(\ref{integer2:constraint-1}) to $n_1(t)$ constraints by specifying $x_i+z_{i}(t)$ to be zero or one for $1 \leq i \leq n_1(t)$, i.e.,  for each packet arriving by slot $t$. In particular, if the value of $y_i$ for all $i$ is known, then  we  specify $x_i+z_{i}(t)=1$ for a total of $n_1(t)-n_2(t)+\sum^{n_2(t)}_{i=1}y_i$ packets arriving at queue~$Q_1$  by slot~$t$ and  $x_i+z_{i}(t)=0$ for all other packets arriving at queue $Q_1$ by slot $t$. There are $n_1(t) \choose n_1(t)-n_2(t)+\sum^{n_2(t)}_{i=1}y_i$ ways to decouple the original constraints.

	\begin{example} \label{ex:counter}
		Suppose that two packets arrive at queue $Q_1$ in slots 1 and 2, respectively, and  one packet arrives at queue $Q_2$ in slot~3.  Assume  transmission cost  $C=4$. In this case, the optimal solution to Problem~\ref{problem2} is $x_{1}=1$ and $x_{2}=0$. Next, given the optimal decision for the  packet at queue~$Q_2$ (i.e., optimally transmitted with coding), we consider constraints for each packet at queue~$Q_1$ as follows: 
		\begin{itemize}
			\item \textit{Slot $t=1$}: A packet arrives at  queue $Q_1$ in slot~$1$. Thus, we can obtain $x_{1}+z_{1}(1)=1$.
			\item \textit{Slot $t=2$}: The other packet arrives at  queue $Q_1$ in slot~$2$.  Thus, we can obtain  $x_{1}+z_{1}(2)=1$ and $x_{2}+z_{2}(2)=1$. 
			\item \textit{Slot $t=3$}:  A packet arrives at queue $Q_2$. Since we are given that the  packet at queue~$Q_2$ optimally codes with a packet at queue~$Q_1$, two   options are following: (1) $x_1+z_1(3)=0$, $x_2+z_2(3)=1$, i.e., the first packet at queue $Q_1$ is  transmitted with coding, and the second packet either is transmitted without coding or waits in slot~$3$; (2) $x_1+z_1(3)=1$, $x_2+z_2(3)=0$.
			\item \textit{Slot $t>3$}: No packet arrives at both queues. Thus,  if  $x_1+z_1(3)=0$ and $x_2+z_2(3)=1$, then $x_1+z_1(t)=0$ and $x_2+z_2(t)=1$; otherwise, $x_1+z_1(t)=1$ and $x_2+z_2(t)=0$.
		\end{itemize}
		We  calculate the minimum  value of $4(x_1+x_2)+\sum_{t=1}^{\infty}z_1(t)+\sum_{t=1}^{\infty}z_2(t)$ subject to the two possible constraints, i.e., forming two different linear programs:
		\begin{itemize}
			\item \textit{Consider the  constraints of $x_1+z_{1}(t) \geq 1$ for $1 \leq t \leq 2$, and $x_2+z_{2}(t) \geq 1$ for $t \geq 2$}: The optimal  solution is $x_1=0$ and $x_2=1$, while the minimum objective value is $6$.
			\item \textit{Consider the constraints of $x_1+z_{1}(t) \geq 1$ for $t \geq 1$, and $x_2+z_{2}(2) \geq 1$}:  The optimal solution is $x_1=1$ and $x_2=0$, while the minimum objective value is $5$.
		\end{itemize}
		Thus, only the second set of constraints is correct.  The  idea underlying   the correct set of constraints is that the first packet waits for a longer time  (for coding) than  the second packet does.
		%We  conclude that, even though the optimal solution for all $y_i$'s is given, the optimal objective value subject to a set of decoupled constraints might be different from that in the original integer program (\ref{interger2}).
		%	\hfill $\blacktriangleleft$
	\end{example} 
	
	Let $\mathbf{I}(t)=\{i: x_i+z_{i}(t) \geq 1\}$ be the set of indices such that the value of $x_i+z_{i}(t)$ in slot $t$ is specified to be greater than or equal to one. Let $\mathbf{I}=\{\mathbf{I}(1), \mathbf{I}(2), \cdots\}$.  Our goal is to identify a \textit{correct} set $\mathbf{I}$ of constraints such that the solution to minimize the cost (in Problem~\ref{problem2})  subject to the set $\mathbf{I}$  is an optimal solution to Problem~\ref{problem2}.  Example~\ref{ex:counter} suggests that, when a packet arrives at queue $Q_2$ in slot $t$, a correct set $\mathbf{I}(t)$ of constraints in slot~$t$ can be obtained by removing the most recent packet in set $\mathbf{I}(t-1)$ of the previous slot.  The argument will be confirmed in the next section.

	\subsection{Primal-dual formulation}\label{subsection:primal-dual-2}
	
	With the idea developed in Example~\ref{ex:counter}, we propose an algorithm in Alg.~\ref{i-alg} for identifying a correct  set $\mathbf{I}$ of constraints. 
	Line~\ref{i-alg:initial} initiates  set $\mathbf{I}(t)$ in slot $t$ to be set $\mathbf{I}(t-1)$ of the previous slot.  When a packet arrives at queue $Q_1$ in slot $t$,  Line~\ref{i-alg:p1} adds the corresponding index to  set $\mathbf{I}(t)$.  Line~\ref{i-alg:n} introduces a variable $q_2$ to indicate the available packets at queue~$Q_2$ for coding; precisely, Line~\ref{i-alg:n} sets the value of  variable $q_2$ to be the present arrivals $A_2(t)$ at queue~$Q_2$. Since Line~\ref{i-alg:j-packet}, if $q_2 \neq 0$ (i.e., there is a packet at queue~$Q_2$) and $\mathbf{I}(t)\neq 0$ (i.e., there is a packet at queue~$Q_1$), then Line~\ref{i-alg:remove2} removes index $i^*$ (i.e., the most recent packet in  set $\mathbf{I}(t)$ as in  Line~\ref{i-alg:optimal-i}) from set $\mathbf{I}(t)$ and Line~\ref{i-alg:n-update} removes one packet from queue~$Q_2$.

	\begin{algorithm}[t]
			\small
		\SetAlgoLined 
		\SetKwFunction{Union}{Union}\SetKwFunction{FindCompress}{FindCompress} \SetKwInOut{Input}{input}\SetKwInOut{Output}{output}
		
		%	\Output{Optimal $\mathbf{I}$ and optimal $b_i$ for all $i$}
		
		\tcc{Initialize  set $\mathbf{I}(t)$ at the beginning of slot $1$ as follows:}	
		$\mathbf{I}(t) \leftarrow \emptyset$ for all $t$\;	
		\tcc{For each slot $t=1, 2, \cdots$, perform as follows:}
		
		%	\While{$t \leq T_{2,N_2}$}{
		$\mathbf{I}(t) \leftarrow \mathbf{I}(t-1)$\;    \label{i-alg:initial}
		
		\ForAll{$i$-th packet arriving at queue $Q_1$ in slot $t$}{
			$\mathbf{I}(t) \leftarrow \mathbf{I}(t) \cup \{i\}$\;   \label{i-alg:p1}
		}
		$q_2 \leftarrow A_2(t)$\;\label{i-alg:n}
		\While{$q_2\neq 0$ and $\mathbf{I}(t) \neq \emptyset$\label{i-alg:j-packet}}{
			
			%			$n_2(t) \leftarrow n_{2,t-1}+1$ \;
			$i^* \leftarrow \max \mathbf{I}(t) $\; \label{i-alg:optimal-i}
			%			\If{$t-T^{(1)}_{i^*} \leq C$}{
			%	$y_j \leftarrow 1$\;	 \label{i-alg:keep}
			%			}
			%			\Else{
			%			%	$y_j \leftarrow 0$\; \label{i-alg:remove1}
			$\mathbf{I}(t) \leftarrow \mathbf{I}(t) - \{i^*\}$\; \label{i-alg:remove2}
			%			}
			$q_2=q_2-1$\; \label{i-alg:n-update}
			
		}
		
		%	}
		\caption{Identifying a correct set $\mathbf{I}$ of constraints}
		\label{i-alg}
	\end{algorithm}

	We formulate a linear program subject to the set $\mathbf{I}$ produced by Alg.~\ref{i-alg} as follows.

	\textbf{Linear program (primal program):} 
	\begin{subequations} \label{primal2}
		\begin{align}  
		\min& \hspace{.5cm}C \cdot \sum^{N_1}_{i=1} x_i+ \sum^{\infty}_{t=1} \sum^{N_1}_{i=1} z_i(t) \label{primal2:objective}\\
		\text{s.t.}& \hspace{.5cm}x_i+z_i(t) \geq 1 \text{\,\,for all\,\,} i \in \mathbf{I}(t) \text{\,\,and\,\,}t. \label{primal2:constraint}
		\end{align}
	\end{subequations}
	
	The next theorem establishes that linear program~(\ref{primal2}) can optimally solve Problem~\ref{problem2}.

	\begin{theorem} \label{theorem:optimal i}
		The solution to linear program~(\ref{primal2}) is an optimal solution to Problem~\ref{problem2}.
	\end{theorem}
	\begin{proof}
		We prove by induction. See Appendix~\ref{appendix:theorem:optimal i} for details.
	\end{proof}
	
	The dual program of primal program (\ref{primal2}) is following. 
	
	\textbf{Dual program:}
	\begin{subequations} \label{dual2}
		\begin{eqnarray} 
		&\max&\sum^{\infty}_{t=1} \sum_{i\in \mathbf{I}(t)} w_i(t) \label{dual2:objective}\\
		&\text{s.t.}& \sum_{t: i \in \mathbf{I}(t)} w_i(t) \leq C \text{\,\,\,for all\,\,} i;\\
		&& 0 \leq w_i(t) \leq 1  \text{\,\,\,for all\,\,} i \text{\,\,and\,\,} t. 
		\end{eqnarray}
	\end{subequations}

	\subsection{Primal-dual  algorithm}\label{subsection:prima-dual-alg2}
	
	Note that Alg.~\ref{i-alg} can \textit{learn}   a correct set $\mathbf{I}(t)$ of constraints for each slot $t$ in the online fashion. Leveraging the online feature, we  develop a  primal-dual  algorithm in Alg.~\ref{pda2} for solving Problem~\ref{problem2} in the online fashion. For each slot $t$, Alg.~\ref{pda2} updates those $x_i$'s   in the set $\mathbf{I}(t)$ in  Lines~\ref{pda2:start} - \ref{pda2:end}. The updating process is similar to that in Alg.~\ref{pda1}.

	\begin{algorithm}[t]
			\small
		\SetAlgoLined 
		\SetKwFunction{Union}{Union}\SetKwFunction{FindCompress}{FindCompress} \SetKwInOut{Input}{input}\SetKwInOut{Output}{output}
		
		%\Input{Bids vector $B$ and side information $H_i$ for all clients}
		%\Output{Encoding matrix $\Gamma$}
		
		\tcc{Initialize all variables at the beginning of slot $1$ as follows:}
		
		$x_i$, $z_i(t)$,  $w_i(t)$ $\leftarrow 0$ for all $i$ and $t$\;
		$\theta \leftarrow (1+\frac{1}{C})^{\lfloor C \rfloor}-1$ \;
		$\mathbf{I}(t) \leftarrow \emptyset$ for all $t$\;	
		%\newline \\
		\tcc{For each new slot $t=1, 2, \cdots$, the variables are updated as follows:}
		%\SetVline

		$\mathbf{I}(t) \leftarrow \mathbf{I}(t-1)$\;    \label{pda2:i-alg:initial}
		
		\ForAll{$i$-th packet arriving at queue $Q_1$ in slot $t$}{
			$\mathbf{I}(t) \leftarrow \mathbf{I}(t) \cup \{i\}$\;   \label{pda2:i-alg:p1}
		}
		
		$q_2 \leftarrow A_2(t)$\;\label{pad2:n}

		\While{$q_2\neq 0$ and $\mathbf{I}(t) \neq \emptyset$\label{pda2:n-update-start}}{
			%			$n_2(t) \leftarrow n_{2,t-1}+1$ \;
			$i^* \leftarrow \max \mathbf{I}(t) $\; \label{pda2:i-alg:optimal-i}
			%		\If{$t-T^{(1)}_{i^*} \leq C$}{
			$\mathbf{I}(t) \leftarrow \mathbf{I}(t) - \{i^*\}$\; \label{pda2:i-alg:remove2}
			%		}
			$q_2=q_2-1$\; \label{pda2:n-update}
			
		}\label{pda2:n-end}

		\ForAll{$i \in \mathbf{I}(t)$}{ \label{pda2:start}
			
			\If{ $x_i <1$}{     
				
				$z_i(t) \leftarrow 1-x_i$\;
				
				$x_i \leftarrow x_i(1+\frac{1}{C})+ \frac{1}{\theta C}$\; \label{pda2:xi}  
				
				$w_i(t) \leftarrow 1$\;

				%	\IF{$\lfloor x \rfloor - \lfloor pre\_x \rfloor \geq 1$} 
				% 		\STATE $tx \leftarrow 1$	
				%	\ENDIF
			} 
			%\Else{ 
			%	 Terminate\;
			%}
		}\label{pda2:end}

		\caption{Primal-dual   algorithm for solving primal program~(\ref{primal2}) and dual program~(\ref{dual2})}
		\label{pda2}
	\end{algorithm}

	Using the  same arguments as those in the proofs of Lemmas~\ref{lemma:feasible-primal1} and \ref{lemma:dual}, the next  lemma  establishes the feasibility of the solution produced by Alg.~\ref{pda2}. 
	
	\begin{lemma}
		Alg.~\ref{pda2} produces  a feasible  solution to  primal program~(\ref{primal2}) and dual program~(\ref{dual2}). 
	\end{lemma}

	Similar to Theorem~\ref{theroem:competitive-ratio1}, the next theorem shows that Alg.~\ref{pda2} can achieve the same competitive ratio as Alg.~\ref{pda1} does. 
	
	\begin{theorem} \label{theorem:comp-ratio-two-side}
		Let $OPT_{(\ref{primal2})}(\mathbf{A})$ be the minimum objective value in linear program~(\ref{primal2}). Then, 
		the primal objective value in Eq.~(\ref{primal2:objective}) computed by Alg.~\ref{pda2} is bounded above by
		\begin{align*}
		(1+\frac{1}{(1+\frac{1}{C})^{\lfloor C \rfloor}-1}) OPT_{(\ref{primal2})}(\mathbf{A}),
		\end{align*}
		for all possible arrival patterns $\mathbf{A}$.	
	\end{theorem}
	\begin{proof}
		See Appendix~\ref{appendix:theorem:comp-ratio-two-side}.
	\end{proof}

	Then,  similar to Alg.~\ref{online-alg1}, we can transform the solution produced by Alg.~\ref{pda2} to a randomized online scheduling algorithm for managing the delay-award coding decision at queue~$Q_1$. In particular, the scheduling algorithm can also achieve the same  expected competitive ratio as that in Theorem~\ref{theorem:expected-competitive ratio1},  approaching $\frac{e}{e-1}$ when cost $C$ is large enough.

	\subsection{Scheduling both queues} \label{subsection:both-queue}
	This section extends Alg.~\ref{pda2} to the case when both queues $Q_1$ and $Q_2$ can wait for each other. In this context, we propose a \textit{waiting-coding queueing system}  consisting of a waiting queue $Q_w$ and a coding  queue $Q_c$ at the relay.  While queue~$Q_w$ stores those packets that can wait for coding,  queue~$Q_c$ stores  those packets that can find coding pairs at the waiting queue immediately upon arrival.

	Precisely, let $Q_w(t)$ and $Q_{c}(t)$ be the number of packets at queue~$Q_w$ and queue~$Q_c$, respectively, at the \textit{end} of slot~$t$. If the $Q_w(t-1)$ packets at queue~$Q_w$ belong to queue~$Q_1$, then  the $A_1(t)$ (i.e., the number of packets arriving at the original queue $Q_1$) new  arriving packets  enter queue~$Q_w$ at the beginning of slot $t$， and
	\begin{enumerate}
		\item if $Q_w(t-1)+A_1(t) \geq A_2(t)$, then  the $A_2(t)$ new arriving packets enter queue~$Q_c$ at the beginning of slot $t$;
		\item if $Q_w(t-1)+A_1(t) < A_2(t)$, then only $Q_w(t-1)+A_1(t)$ out of the $A_2(t)$ new arriving packets enter queue~$Q_c$ at the beginning of slot $t$, but the remaining $A_2(t)-(Q_w(t)+A_1(t))$ packets enter queue~$Q_w$ at the beginning of slot $t$. 
	\end{enumerate} 
	In contrast, if the $Q_w(t-1)$ packets at queue~$Q_w$ belong to queue~$Q_2$, then the waiting-coding queueing system operates in the opposite way. In other words, while  packets entering queue~$Q_c$ are transmitted (with coding) immediately upon arrival,  packets entering queue~$Q_w$ need scheduling decisions. With the transformation, the waiting-coding queueing system becomes the previously discussed  model where only packets at queue~$Q_w$ can wait for coding. Thus, the randomized online scheduling algorithm associated with Alg.~\ref{pda2} can apply to the waiting-coding queueing system with the expected competitive ratio in Theorem~\ref{theorem:expected-competitive ratio1}.  Furthermore, Section~\ref{section:simulation} will demonstrate the superiority of the proposed scheduling algorithm and the proposed waiting-coding queueing system via computer simulations. 
	
	%Then, the queueing dynamics is   and 
	%\begin{align*}
	%&Q_w(t+1)\\
	%=&\left\{
	%\begin{array}{ll}
	%\max\{(Q_w(t)+A_1(t))-A_2(t), A_2(t)-(Q_w(t)+A_1(t))\} & \text{if the packet at $Q_w(t)$ belongs to $Q_1$;}\\
	%\max\{(Q_w(t)+A_2(t))-A_1(t), A_1(t)-(Q_w(t)+A_2(t))\} & \text{if the packet at $Q_w(t)$ belongs to $Q_2$,}
	%\end{array}
	%\right.
	%\end{align*}
	%for all $t$. The term  $Q_w(t)+A_1(t)$ in the first case of the above equation means the total number of packets at $Q_1$ at the beginning of slot $t$, and $A_2(t)$ is the total number of packets at $Q_2$

	%However, even though  $y_i=1$ for some $i$ is known, it is unclear if any packet arriving at $Q_1$ prior to $T^{(2)}_i$ wants to wait for coding with the $i$-th packet at $Q_2$,  since  the $i$-th packet at $Q_2$ might desire to wait for a future packet at $Q_1$.  Without a clear idea like Alg.~\ref{i-alg} to decouple the original integer program,  Alg.~\ref{pda2} cannot be extended immediately.  

	\subsection{A transmission constraint} \label{subsection:constraint}
	Recall that Alg.~\ref{pda2} might transmit more than one packet in a slot (but less than three uncoded packets, as shown in Lemma~\ref{lemma:number-of-tx}). This Section considers a transmit constraint: the relay can transmit at most one packet in each slot. According to Section~\ref{subsection:both-queue}, we can focus on scheduling  packets at queue~$Q_1$ while  all  packets at queue~$Q_2$ are transmitted immediately upon arrival. 
	
	Under the transmission  constraint,  if more than one packet arrive at a queue, then those additional packets (except for  one of them) cannot be processed in the arriving slot for any scheduling algorithm. Thus, without loss of generality, we can  further assume that at most one packet can arrive at each queue in each slot. If more than one packet arrives at a queue, we can just move them to the following slots, so that at most one packet arrives at that queue. With that assumption, we analyze the number of \textit{uncoded} packets  required by the randomized online scheduling algorithms (like Alg.~\ref{online-alg1}) associated with  Alg.~\ref{pda2}: following the proof of Lemma~\ref{lemma:number-of-tx}, the number of uncoded packets transmitted in slot~$t$ is
	\begin{align*}
	\sum_{i \in \mathbf{I}(t)}\Delta \widetilde{x}_i(t) \leq \sum^{\lfloor C \rfloor}_{j=1} \frac{1}{\theta \cdot C}(1+\frac{1}{C})^j=1,
	\end{align*}
	where the inequality is because: (1) at most $\lfloor C \rfloor$ packets (as in the proof of Lemma~\ref{lemma:dual}) in set~$\mathbf{I}(t)$ that can be updated by Line~\ref{pda2:xi} of Alg.~\ref{pda2} in slot~$t$; (2) the $j$-th most recent packet in set~$\mathbf{I}(t)$ has been updated by Line~\ref{pda2:xi} of Alg.~\ref{pda2} for at least $j$ times since its arrival; (3) the value of $\Delta \widetilde{x}_i(t)$ is  $\frac{1}{\theta \cdot C}(1+\frac{1}{C})^j$ if the $i$-th packet is updated by Line~\ref{pda2:xi} of Alg.~\ref{pda2} for $j$ times. 
	
	We emphasize that, by the above analysis, the randomized online scheduling algorithm might  need two transmissions in a slot, i.e., one potential coded packet plus one potential uncoded packet. To make the randomized online scheduling algorithm perform under the constraint of at most one transmission, we revise Alg.~\ref{pda2} as follows: the updates in Lines~\ref{pda2:start} - \ref{pda2:end} perform only when no packet arrives at queue~$Q_2$. That is because, if a packet arrives at queue~$Q_2$, the relay has to transmit a coded packet; thus, stop updating those variables for transmitting an uncoded packet. Following the line in \cite[Theorem~5]{tseng2019online}, the randomized online scheduling algorithm associated with the revised Alg.~\ref{pda2} can also achieve the same expected competitive ratio of $\frac{e}{e-1}$ when cost~$C$ is large enough. Moreover, Section~\ref{section:simulation} will  validate  the revised randomized online scheduling algorithm  via computer simulations. 
	
	\begin{remark}\label{remark:on-off-channel}
		We remark that the revised randomized online scheduling algorithm can also solve the adversarial ON-OFF channel, also by stopping updating when the channel is OFF. 
	\end{remark}

	\section{Numerical studies}\label{section:simulation}
	We have analyzed the proposed randomized online scheduling algorithm in the worst-case scenario; in contrast, we investigate the proposed  algorithm  in  the average-case  scenario by computer simulations in this section. 
	
	First, we simulate  a single-relay network (as in Fig.~\ref{fig:network}-(a)) where  packets arrive at queues~$Q_1$ and~$Q_2$ according to  the i.i.d.  Bernoulli distributions with means  $\max\{\min\{P_1,1\},0\}$ and $\max\{\min\{P_2,1\},0\}$, respectively, where $P_1$ and $P_2$ are the Gaussian random variables (for adding some noises to the Bernoulli arrivals) with means $p_1$ and $p_2$, respectively, and variance $\sigma^2$. Moreover, the relay can transmit at most one packet for each slot. We compare the proposed scheduling algorithm (i.e., the randomized online scheduling algorithm associated with Alg.~\ref{pda2} along with the waiting-coding queueing system in Section~\ref{subsection:both-queue} and the stopping mechanism in Section~\ref{subsection:constraint}) with  \textit{threshold-type scheduling algorithms}, where the relay transmits an uncoded packet in a slot if (in the original queueing system) a queue is empty and the non-empty queue size is over its threshold in that slot.  The \textit{optimized-threshold scheduling algorithm} was proposed in \cite{hsu2014opportunities}  for minimizing the long-run average cost in the stochastic environment. However, deriving an optimal threshold for each queue needs the statistics  $p_1$ and $p_2$, i.e., the optimized-threshold scheduling algorithm is an offline scheduling algorithm. Fig.~\ref{fig:sim-cost} displays the ratio between the total  cost (in 10,000 slots) incurred by the proposed scheduling algorithm and that incurred by the optimized-threshold scheduling algorithm.  We can observe  that the ratio for the proposed scheduling algorithm is at most 1.35 (in Fig.~\ref{fig:sim-cost}-(a) when $p_2=0.1$ and $C=10$).  That is, the proposed algorithm performs much better than what we analyzed in the worst-case scenario (with the expected competitive ratio of $\frac{e}{e-1}\approx1.58$).   In addition, Fig.~\ref{fig:sim-cost} also displays the ratio between the total cost incurred by the \textit{$C$-threshold scheduling algorithm} and that incurred by the optimized-threshold scheduling algorithm, where the $C$-threshold is an online scheduling algorithm with the constant threshold $C$ and was analyzed in \cite{ciftcioglu2011cost}. According to Fig.~\ref{fig:sim-cost}, our algorithm significantly outperforms the $C$-threshold scheduling algorithm.  Moreover, We can observe that the ratio (for a fixed $p_2$ and a fixed $C$) decreases as the variance increases. That is because the ratio in Fig.~\ref{fig:sim-cost}-(a) decreases when the expected arrival rate $p_2$ at queue $Q_2$ moves toward 0.5  and  the expected arrival rate at queue~$Q_2$ in  Figs.~\ref{fig:sim-cost}-(b) and~\ref{fig:sim-cost}-(c) (i.e., $E[\max\{\min\{P_2,1\},0\}]$) moves toward 0.5 (because of the truncation of the Gaussian variable $P_2$ to 0 and 1) when the variance increases.

%	\begin{figure}[!t]
%		\begin{minipage}{.33\textwidth}
%			\centering
%			\includegraphics[width=\textwidth]{cost-sigma0.eps}
%		\end{minipage}
%	\end{figure}
	
	\begin{figure}[!t]
		\begin{minipage}{.33\textwidth}
			\centering
			\includegraphics[width=\textwidth]{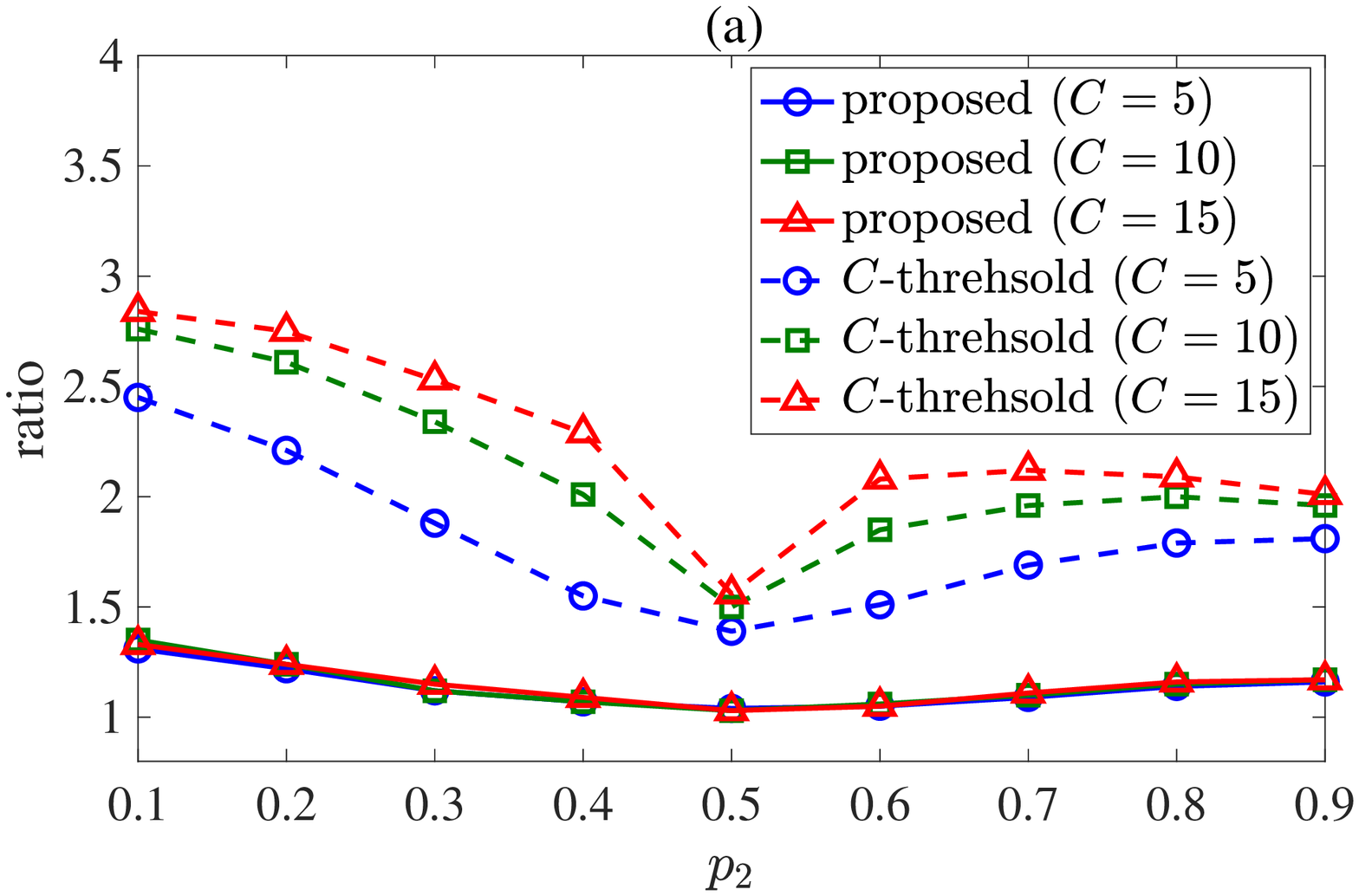}
		\end{minipage}\hfill
		\begin{minipage}{.33\textwidth}
			\centering
			\includegraphics[width=\textwidth]{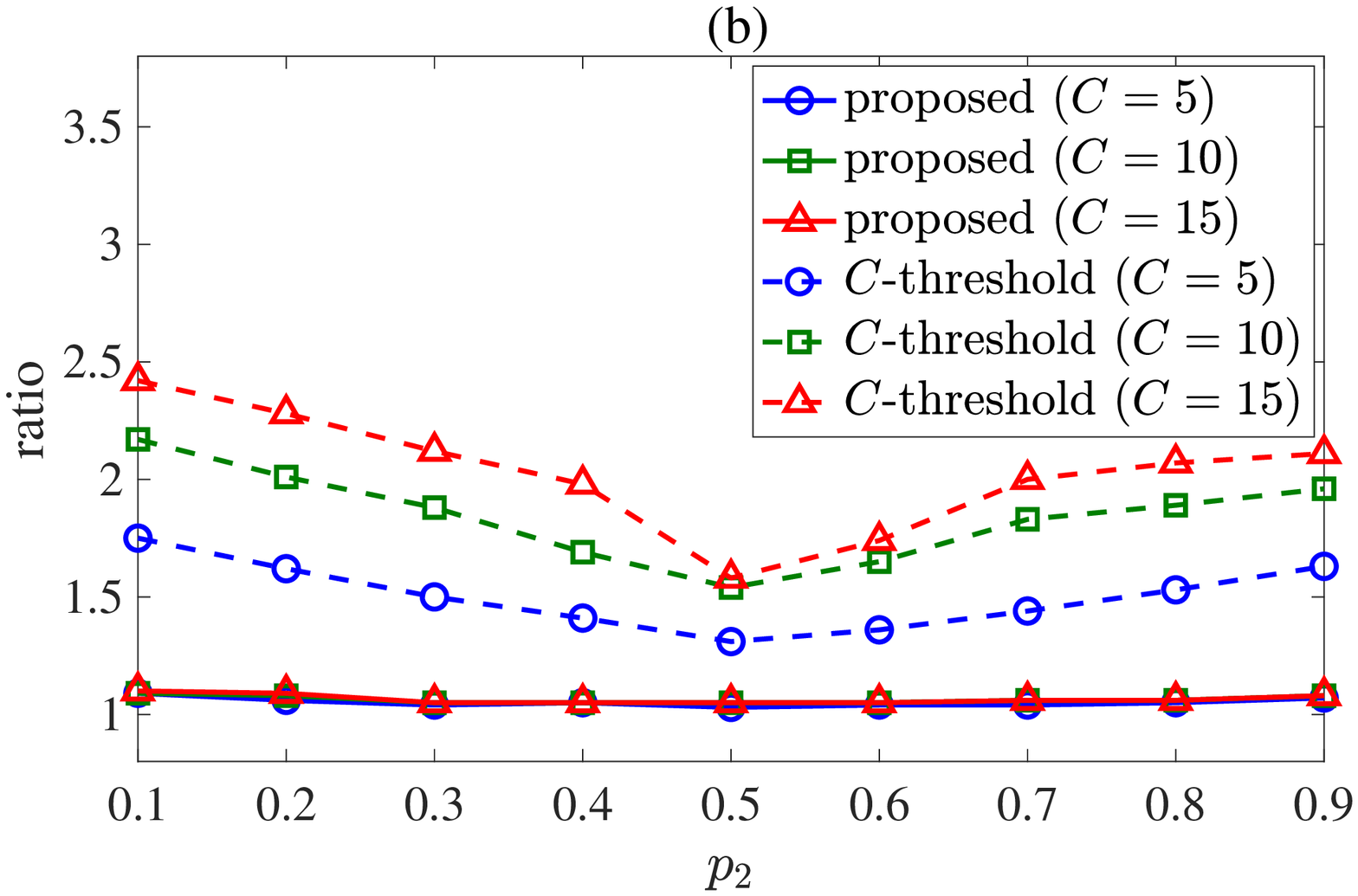}
		\end{minipage} \hfill
		\begin{minipage}{.33\textwidth}
			\centering
			\includegraphics[width=\textwidth]{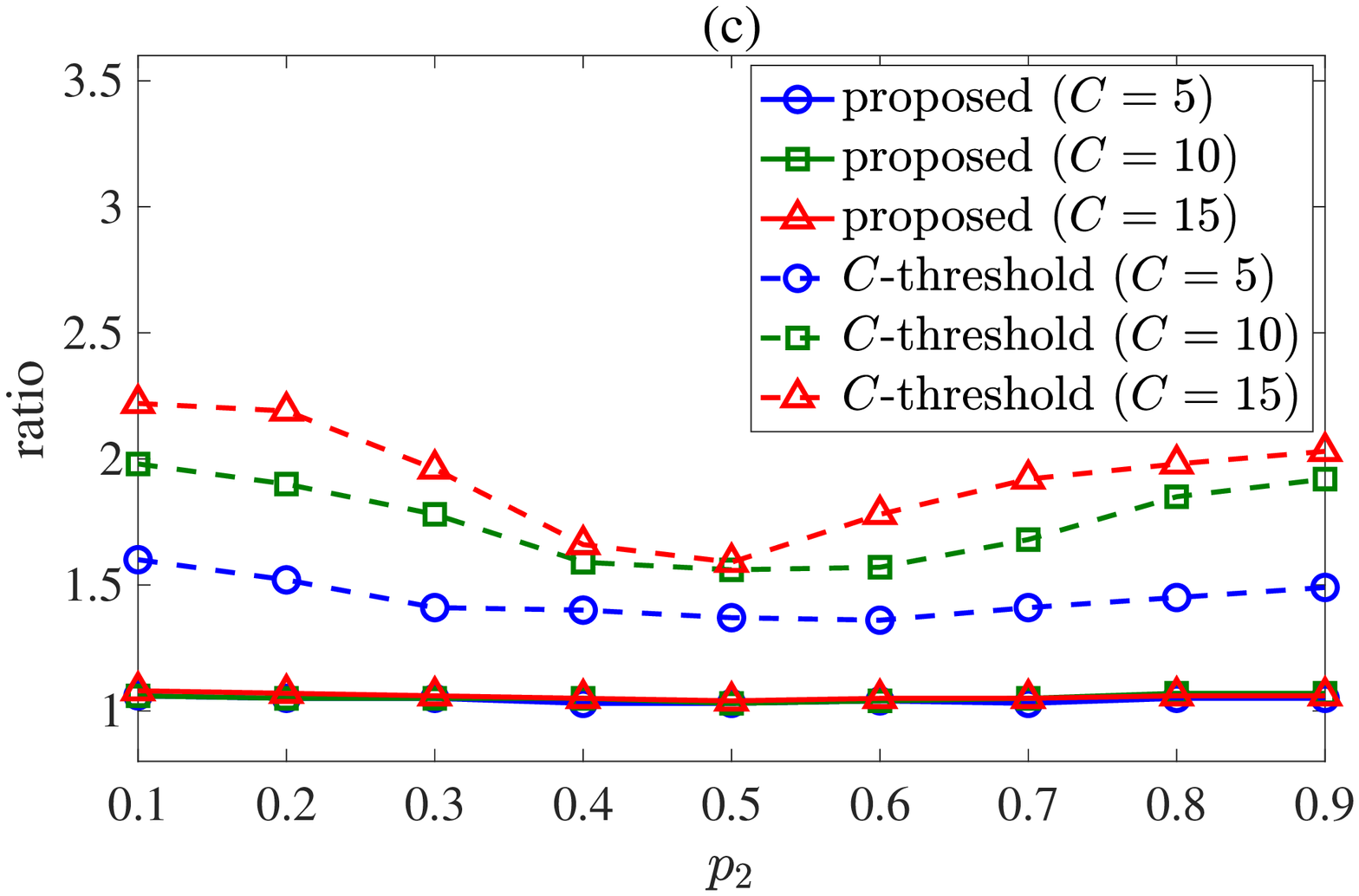}
		\end{minipage}\hfill
		\caption{Ratio versus  $p_2$ (fixed $p_1=0.5$) in the single-relay network: (a) $\sigma^2=0$; (b) $\sigma^2=1$; (c) $\sigma^2=2$.}
		\label{fig:sim-cost}
	\end{figure}

Second, we investigate the coding overheads incurred by the three scheduling algorithms.  Fig~\ref{fig:sim-coding} displays the number of coded packets when $\sigma^2=0$ (i.e., for the case in Fig.~\ref{fig:sim-cost}-(a)). We can observe that while the proposed scheduling algorithm yields less coded packets than the optimized-threshold scheduling algorithm,  the $C$-threshold scheduling algorithm yields more than that. That is, while the proposed scheduling algorithm is a little conservative (in waiting for coding), the $C$-threshold scheduling algorithm waits too long. That is why the proposed scheduling algorithm and the $C$-threshold scheduling algorithm cannot minimize the total cost.

	\begin{figure}[!t]
	\begin{minipage}{.33\textwidth}
		\centering
		\includegraphics[width=\textwidth]{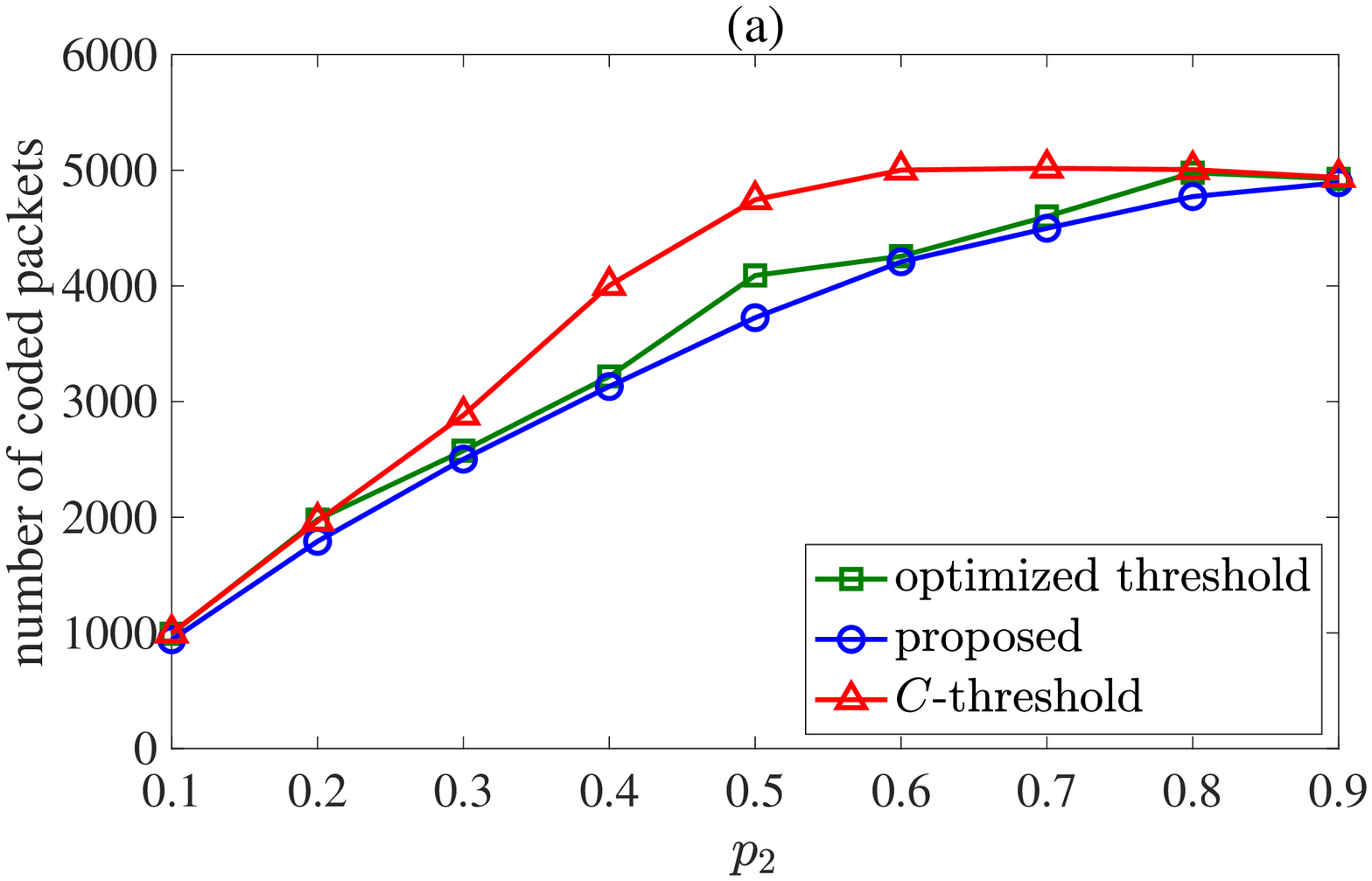}
	\end{minipage}\hfill
	\begin{minipage}{.33\textwidth}
		\centering
		\includegraphics[width=\textwidth]{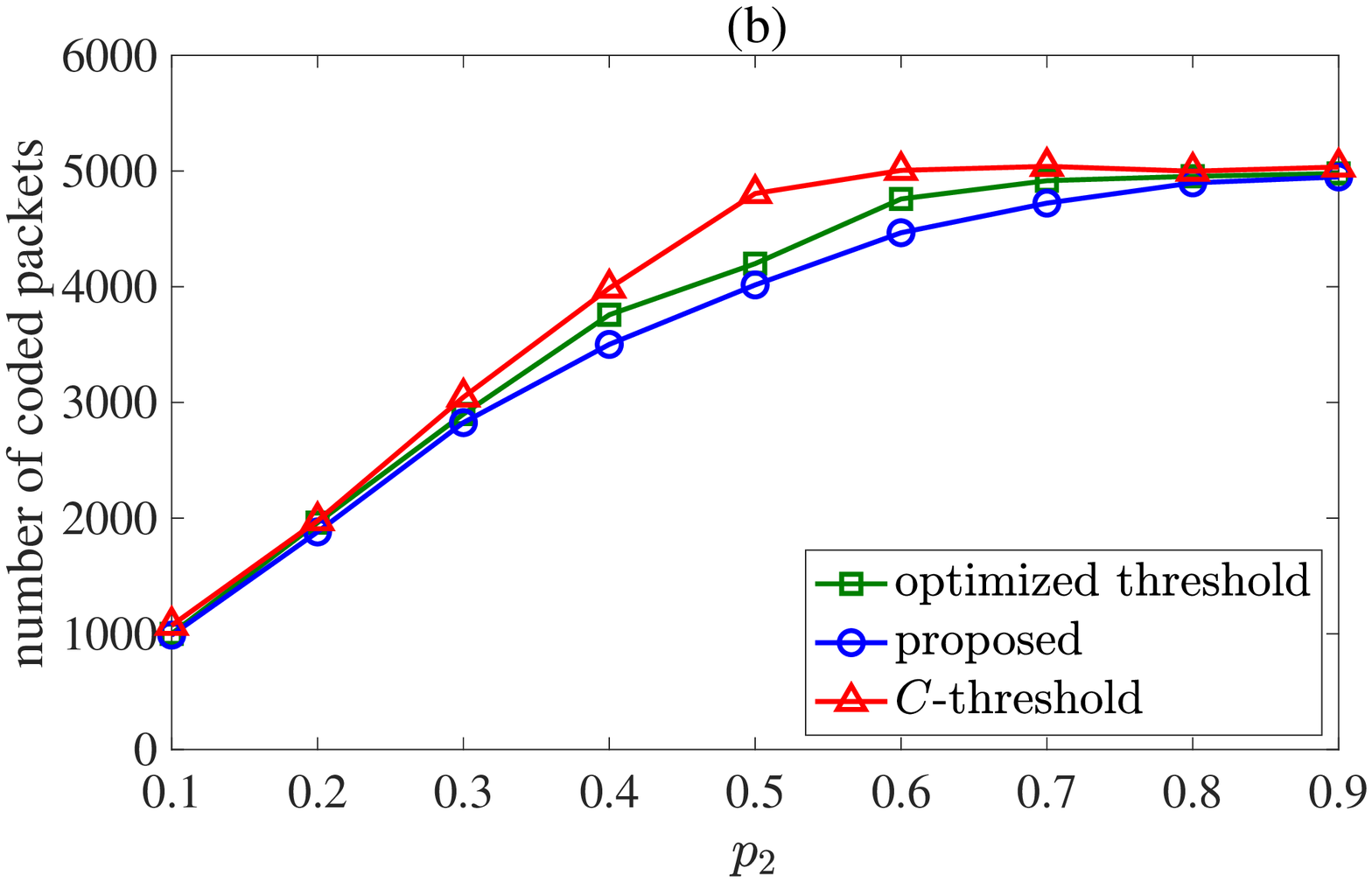}
	\end{minipage} \hfill
	\begin{minipage}{.33\textwidth}
		\centering
		\includegraphics[width=\textwidth]{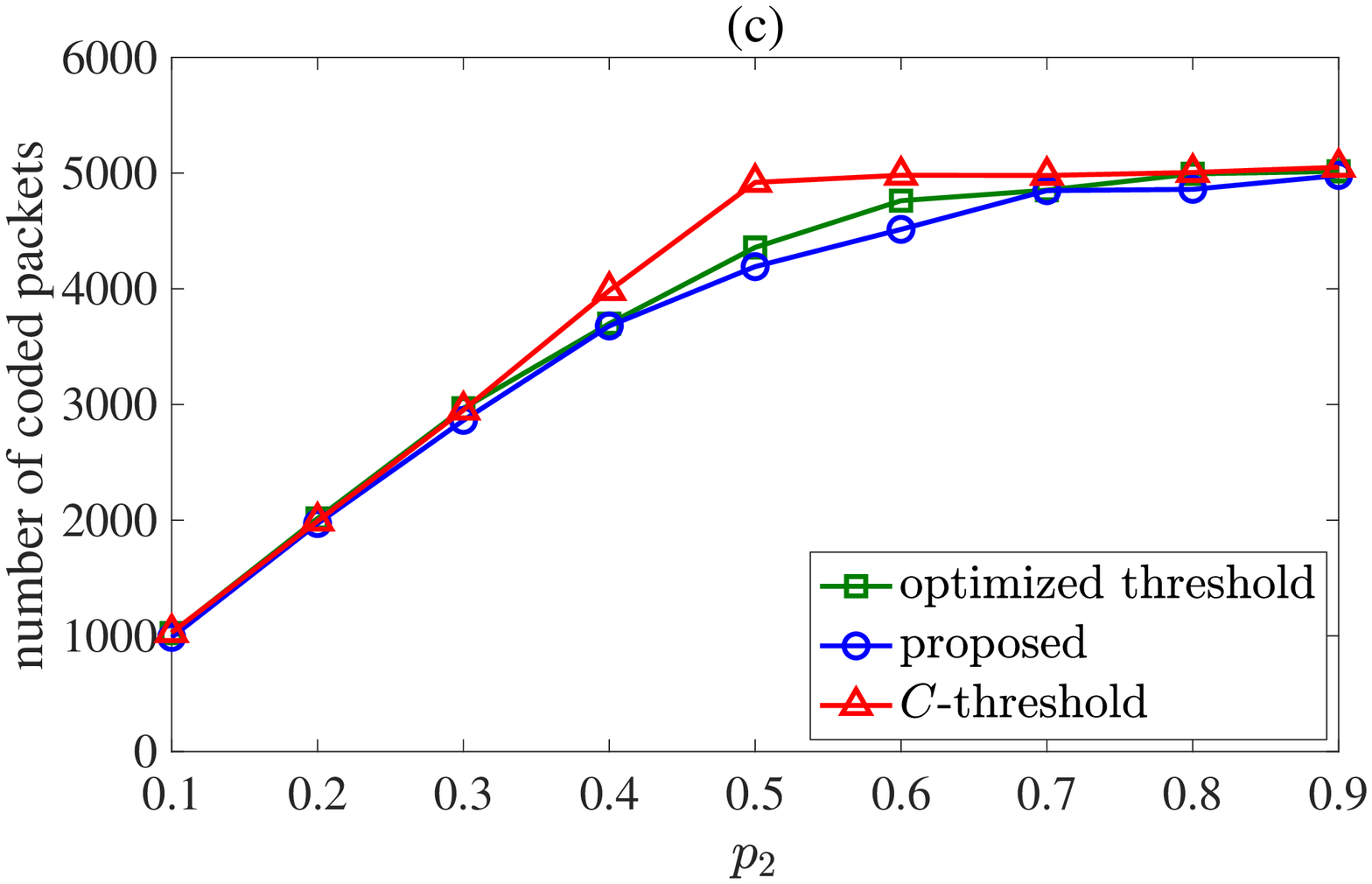}
	\end{minipage}\hfill
	\caption{Number of coded packets versus $p_2$ (fixed $p_1=0.5$ and $\sigma^2=0$) in the single-relay network: (a) $C=5$; (b) $C=10$; (c) $C=15$.}
	\label{fig:sim-coding}
\end{figure}

	Third, we simulate multi-relay networks where external packets arrive at the two end relays according to the i.i.d. Bernoulli distributions. While each relay can transmit at most one packet for each slot, a received packet from the other relay in a slot cannot be processed until the next slot. 
	Fig.~\ref{fig:more-relay}-(a) displays the ratio of total costs (with respect to the optimized-threshold  scheduling algorithm) when there are two relays and both relays take transmission costs $C_1$ and $C_2$, respectively. The optimized-threshold scheduling algorithm identifies a threshold for each queue by exhaustive search for minimizing the total cost among all possible thresholds.  We want to emphasize that an optimal scheduling for the two-relay network is still unclear. In particular, the optimized-threshold scheduling algorithm might not minimize the long-run average cost in this case, though its great performance has been demonstrated in \cite{hsu2014opportunities} by computer simulations. We can observe that the ratios for the proposed scheduling algorithm and the $C$-threshold scheduling algorithm in Fig.~\ref{fig:sim-cost}-(a) and Fig.~\ref{fig:more-relay}-(a) are almost the same. In addition, Fig.~\ref{fig:more-relay}-(a) also displays the ratio for the sub-optimized-threshold scheduling algorithm, which identifies a threshold for each queue by exhaustive search for minimizing the total cost subject to the condition that all left queues have the same threshold and all right queues do as well. We can observe that  the sub-optimized-threshold scheduling algorithm can achieve almost the same total cost as  the optimized-threshold scheduling algorithm does. Thus, we compared the proposed scheduling algorithm with the sub-optimized-threshold scheduling algorithm when there are more than two relays and all relays take the same transmission cost~$C$. Figs.~\ref{fig:more-relay}-(b) and Figs.~\ref{fig:more-relay}-(c) displays the ratio with respect to the sub-optimized-threshold scheduling algorithm. We can observe that  the ratio is insensitive to the numbers of relays; in particular, the proposed scheduling algorithm still significantly outperforms the $C$-threshold scheduling algorithm. 
	
		\begin{figure}[!t]
		\begin{minipage}{.33\textwidth}
			\centering
			\includegraphics[width=\textwidth]{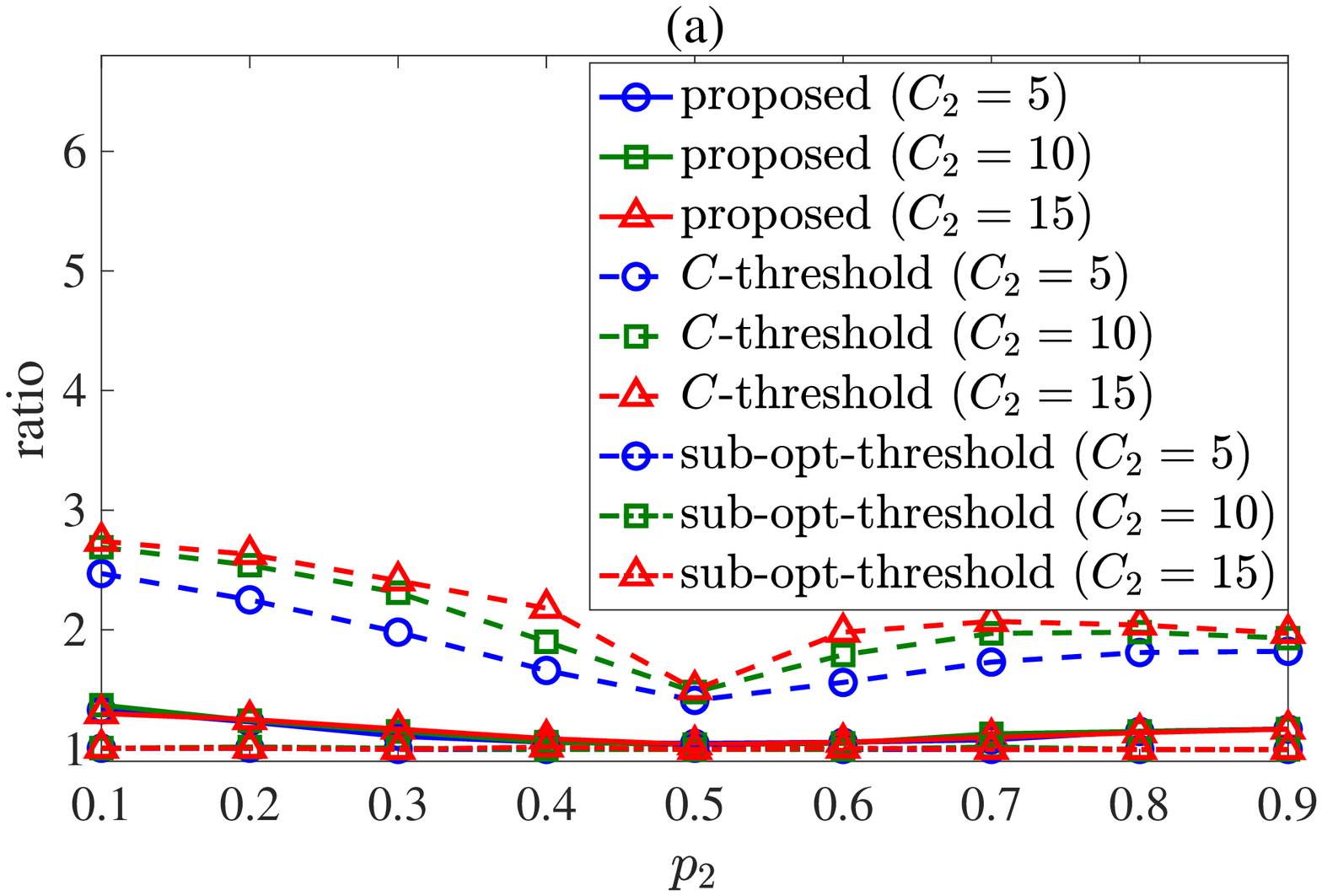}
		\end{minipage}\hfill
		\begin{minipage}{.33\textwidth}
			\centering
			\includegraphics[width=\textwidth]{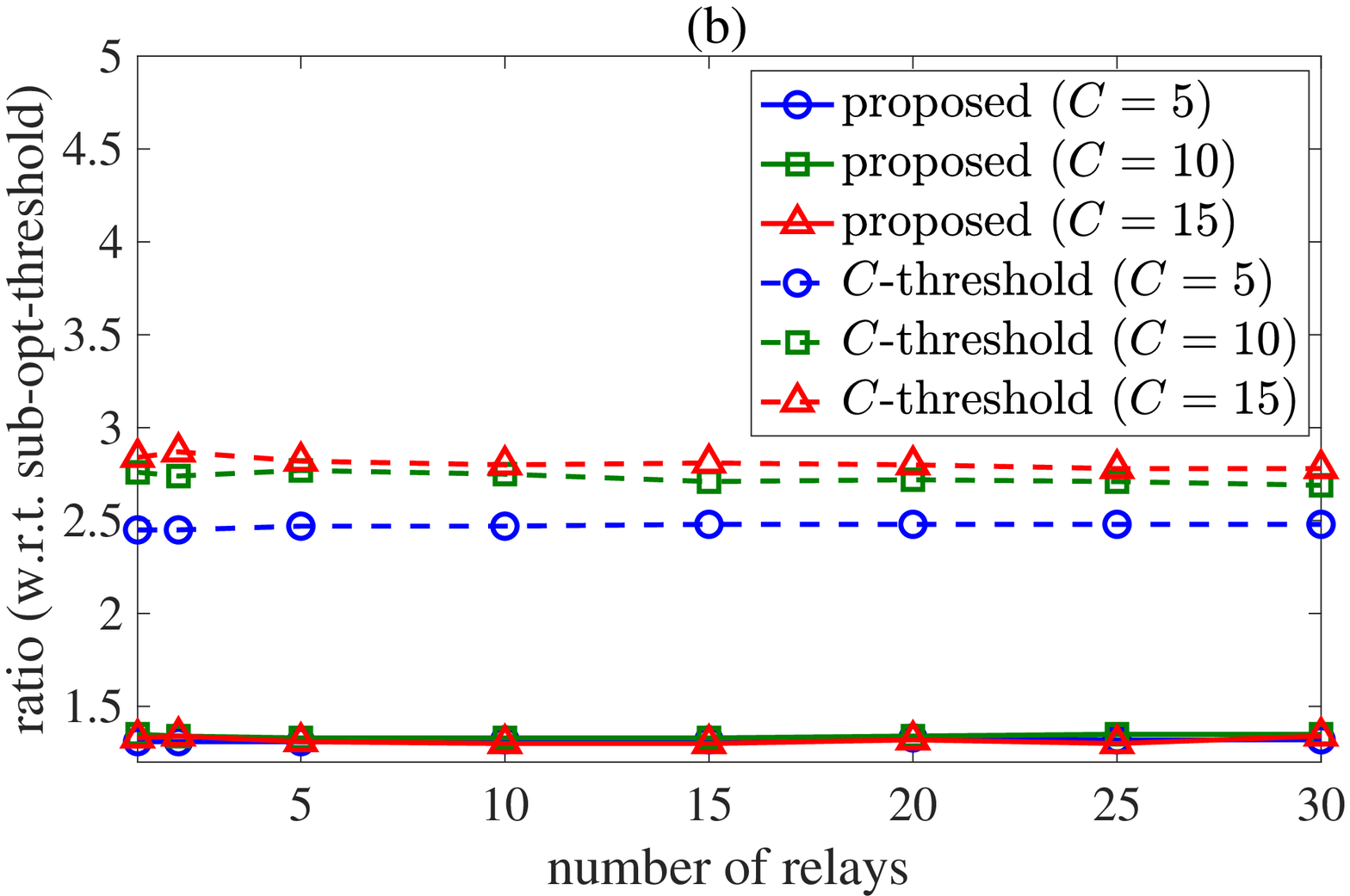}
		\end{minipage} \hfill
		\begin{minipage}{.33\textwidth}
			\centering
			\includegraphics[width=\textwidth]{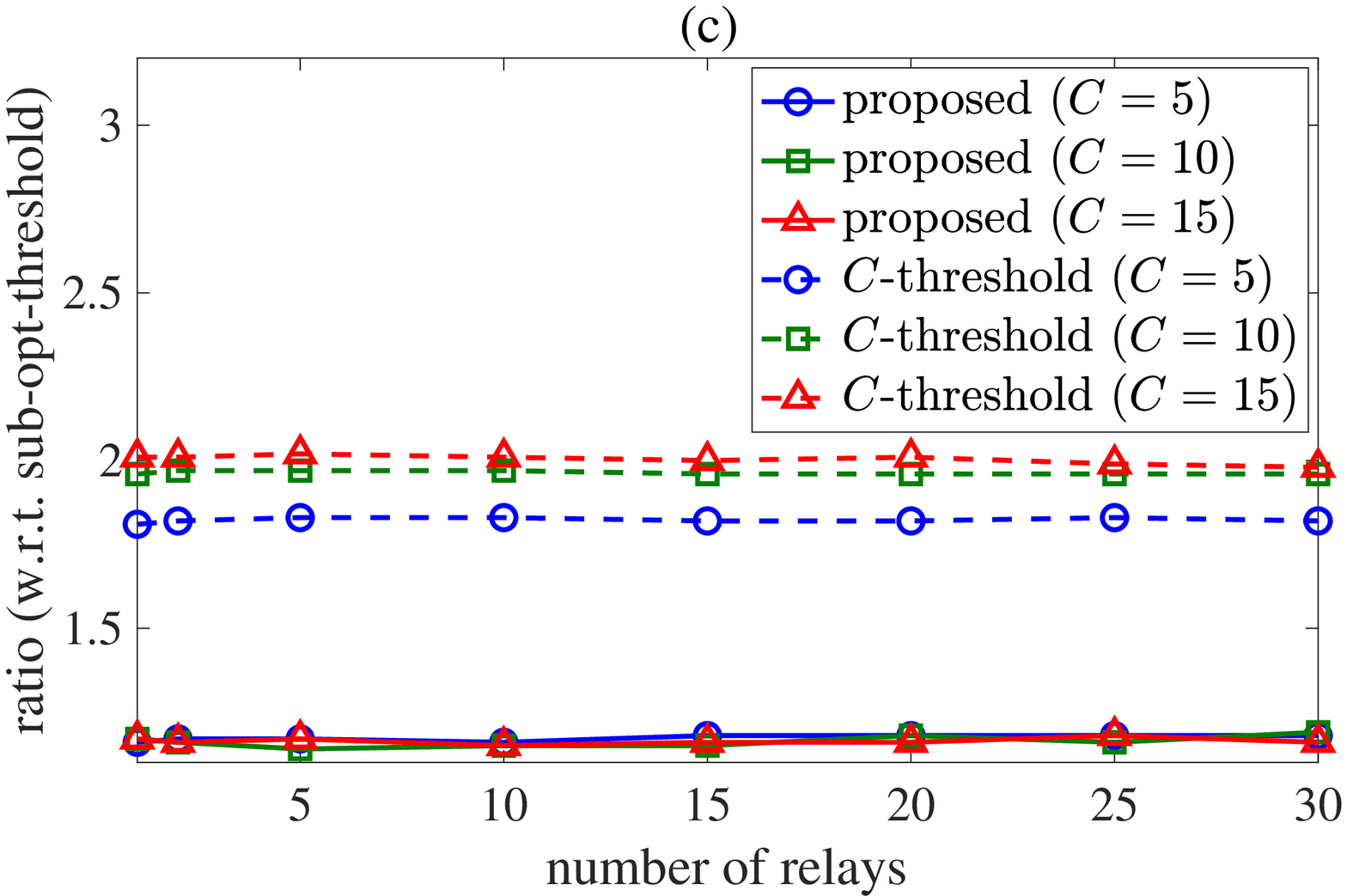}
		\end{minipage}\hfill
		\caption{(a) Ratio versus $p_2$ (fixed $p_1=0.5$, $\sigma^2=0$ and $C_1=5$) when there are two relays; (b) Ratio with respect to the sub-optimized-threshold scheduling algorithm versus the number of relays (fixed $p_1=0.5$, $p_2=0.1$, $\sigma^2=0$); (c) Ratio with respect to the sub-optimized-threshold scheduling algorithm versus the number of relays (fixed $p_1=0.5$, $p_2=0.9$, $\sigma^2=0$).}
		\label{fig:more-relay}
	\end{figure}

%	\begin{figure}[!t]
%		\begin{minipage}{.33\textwidth}
%			\centering
%			\includegraphics[width=\textwidth]{result4.eps}
%		\end{minipage}\hfill
%		\begin{minipage}{.33\textwidth}
%			\centering
%			\includegraphics[width=\textwidth]{result5.eps}
%		\end{minipage} \hfill
%		\begin{minipage}{.33\textwidth}
%			\centering
%			\includegraphics[width=\textwidth]{result6.eps}
%		\end{minipage}\hfill
%		\caption{Average cost versus average arrival rate $p_2$ (fixed $p_1=0.5$) in a two-relay network (fixed cost $C_1=5$): (a) cost $C_2=5$; (b) $C_2=10$; (c) $C_2=15$.}
%		\label{fig:sim2}
%	\end{figure}
%	
	
	\section{Concluding remarks}
	In this paper, we treated a wireless line network employing wireless network coding.  The inherent trade-off between packet delays and  transmission power consumption under adversarial traffic was studied. In particular, we developed a randomized online scheduling algorithm. The proposed scheduling algorithm not only can theoretically guarantee the expected competitive ratio of $\frac{e}{e-1}\approx1.58$ for each relay, but also can numerically approach the minimum total cost (including a delay cost and a transmission cost) by computer simulations; moreover,  the proposed scheduling algorithm can  solve more general ski-rental settings. 
	
	While this paper focused on line networks, some discussions on extending to more general networks are following. Consider a relay with multiple line networks traversing, where the two end nodes of each line network exchange their packets. If the relay can transmit one packet in each slot for each line network (as in Section~\ref{subsection:constraint}), then the proposed scheduling algorithm can immediately apply to each line network individually. However, if a relay has a transmission constraint on the total number of transmissions for all line networks, then linear program~(\ref{primal2}) needs another constraint for specifying that the total number of transmissions cannot be over that transmission constraint. That is an interesting future work. To solve the problem, the prior work \cite{buchbinder2009online} (considering a ``box constraint'') might be helpful.
	
	Some problems are still open as follows. This paper focused on the worst-case analysis. To theoretically analyze the proposed algorithm in the average-case scenario is interesting and can help understand why it has a great performance in the simulation results. Moreover, we analyzed the competitive ratio of the proposed algorithm in the waiting-coding queueing system; however, the competitive ratio in the original queueing system is still undiscovered. Finally, a MAC protocol is given to this paper. Joint scheduling design of MAC and coding would be a promising future topic. 
	
	\appendices

	\section{Proof of Lemma \ref{lemma:integrality}} \label{appendix:lemma:integrality}
	Let $x^*$ and $z^*(t)$ be an optimal solution to  linear program~(\ref{primal1}). Note that $z^*(t)=\max\{N_1-n_2(t)-x^*, 0\}$ by Eqs.~(\ref{primal1:constraint-1}) and~(\ref{primal1:constraint-2}).  Let $T=\max \{t: N_1-n_2(t)-x^*> 0\}$.
	Suppose that the optimal solution is fractional. We  prove  by contradiction, according to the following two cases. 
	\begin{enumerate}
		\item \textit{If cost $C \geq T$}: Write $x^*=\lfloor x^* \rfloor+\epsilon$ with  $\epsilon >0$. Then, the minimum objective value in Eq.~(\ref{primal1:objective}) is
		\begin{align*}
			&C(\lfloor x^* \rfloor+\epsilon)+\sum^{T}_{t=1}N_1-n_2(t)-\lfloor x^* \rfloor-\epsilon\\
			=&C \cdot \lfloor x^* \rfloor+ \left(\sum^{T}_{t=1} N_1-n_2(t)-\lfloor x^* \rfloor\right) + \epsilon (C -T)\\
			\geq& C \cdot \lfloor x^* \rfloor+\sum^{T}_{t=1} N_1-n_2(t)-\lfloor x^* \rfloor.
		\end{align*}
		That is, the solution $x=\lfloor x^* \rfloor$ and $z(t)=\max\{N_1-n_2(t)-\lfloor x^* \rfloor, 0\}$  can produce a smaller objective value in Eq.~(\ref{primal1:objective}) than the solution $x=x^*$ and $z(t)=z^*(t)$ does.
		\item \textit{If cost $C \leq T$}: Write $x^*=\lceil x^* \rceil-\epsilon$ with $\epsilon >0$. Then, the minimum objective value in Eq.~(\ref{primal1:objective}) is
		\begin{align*}
			&C(\lceil x^* \rceil-\epsilon)+\sum^{T}_{t=1}N_1-n_2(t)-\lceil x^* \rceil+\epsilon\\
			=&C \cdot \lceil x^* \rceil+\left(\sum^{T}_{t=1} N_1-n_2(t)-\lceil x^* \rceil\right) + \epsilon (T-C)\\
			\geq& C \cdot \lceil x^* \rceil+\sum^{T}_{t=1} N_1-n_2(t)-\lceil x^* \rceil.
		\end{align*}
		That is, the solution $x=\lceil x^* \rceil-\epsilon$ and $z(t)=\max\{N_1-n_2(t)-\lceil x^* \rceil, 0\}$  can produce a smaller objective value in Eq.~(\ref{primal1:objective}) than the solution $x=x^*$ and $z(t)=z^*(t)$ does.
	\end{enumerate}
	By these contradictions, we conclude that the optimal solution to the linear program is integral.

	\section{Proof of Lemma~\ref{lemma:feasible-primal1}}\label{appendix:lemma:feasible-primal1}
	First, the primal constraint  in Eq.~(\ref{primal1:constraint-2}) holds obviously because Alg.~\ref{pda1} initializes all variables to be zeros in Lines~\ref{pda1:initial} and \ref{pda1:more-variable} and never decreases their values. Second,   the primal constraint in Eq.~(\ref{primal1:constraint-1}) holds for each slot $t$ as follows:
	\begin{enumerate}
		\item If $\widehat{x}_i(t) <1$ for $i=n_2(t)+1, \cdots, N_1$,  then Line~\ref{pda1:zi} of Alg.~\ref{pda1} yields 
		\begin{align}
		\widetilde{x}_i(t)+\widetilde{z}_i(t)=\widetilde{x}_i(t)+(1-\widehat{x}_i(t)) \geq 1, \label{eq:x-z-inequal}
		\end{align}
		where the inequality is based on $\widetilde{x}_i(t) > \widehat{x}_i(t)$ as Alg.~\ref{pda1} increases the value of $x_i$ in Line~\ref{pda1:xi}. Thus, the solution produced by Alg.~\ref{pda1} satisfies the primal constraint in Eq.~(\ref{primal1:constraint-1}) since
		\begin{align*}
		\widetilde{x}(\infty)+\widetilde{z}(t)\geq  \widetilde{x}(t)+\widetilde{z}(t)=&\sum_{i=1}^{N_1}\left(\widetilde{x}_i(t)+\widetilde{z}_i(t)\right)\\
		\mathop{\geq}^{(a)}&  \sum_{i=n_2(t)+1}^{N_1} \left(\widetilde{x}_i(t)+\widetilde{z}_i(t)\right)\mathop{\geq}^{(b)}  \sum_{i=n_2(t)+1}^{N_1} 1 = N_1-n_2(t),
		\end{align*}
		where  (a) is due to the values of the variables are non-negative; (b) is based on Eq.~(\ref{eq:x-z-inequal}).
		\item if $\widehat{x}_i(t) \geq 1$ for $i=n_2(t)+1, \cdots, N_1$, then the solution produced by Alg.~\ref{pda1} satisfies the primal constraint in Eq.~(\ref{primal1:constraint-1}) as well since
		\begin{align*}
		\widetilde{x}(\infty)+\widetilde{z}(t)\geq \widetilde{x}(t)+\widetilde{z}(t) \geq  \sum_{i=n_2(t)+1}^{N_1} \widehat{x}_i(t) \mathop{\geq}^{(a)} N_1-n_2(t),
		\end{align*}
		where (a) is due to $\widehat{x}_i(t) \geq 1$, for $i=n_2(t)+1, \cdots, N_1$, in this case.  
	\end{enumerate}
	Then, we complete the proof.

	\section{$\Delta \mathscr{P}(t)$ and $\Delta \mathscr{D}(t)$ in the proof of Theorem~\ref{theroem:competitive-ratio1}}\label{appendix:theroem:competitive-ratio1}
	We derive $\Delta \mathscr{P}(t)$ and $\Delta \mathscr{D}(t)$ as follows.
	\begin{enumerate}
		\item If $\widehat{x}_i(t) <1$ for $i=n_2(t)+1, \cdots, N_1$,  then  $\Delta \mathscr{P}(t)$ can be expressed as
		\begin{align*}
		\Delta \mathscr{P}(t)=&C \cdot (\widetilde{x}(t)-\widehat{x}(t))+\widetilde{z}(t)\\
		\mathop{=}^{(a)}&\sum_{i=n_2(t)+1}^{N_1} C \cdot\left(\widetilde{x}_i(t)-\widehat{x}_i(t)\right) + \widetilde{z}_i(t)\\
		\mathop{=}^{(b)}&\sum_{i=n_2(t)+1}^{N_1} C \cdot \left(\frac{\widehat{x}_i(t)}{C}+\frac{1}{\theta \cdot C}\right) + (1-\widehat{x}_i(t))\\
		=&(N_1-n_2(t))\left(1+\frac{1}{\theta}\right),
		\end{align*}
		where (a) is based on Lines~\ref{pda1:z} and~\ref{pda1:x} of Alg.~\ref{pda1}; (b) is based on Lines~\ref{pda1:zi} and \ref{pda1:xi} of Alg.~\ref{pda1}. 
		Moreover, $\Delta \mathscr{D}(t)=N_1-n_2(t)$ since Alg.~\ref{pda1} updates $w(t)$ to be one in Line~\ref{pda1:y}.
		
		\item If $\widehat{x}_i(t) \geq 1$ for $i=n_2(t)+1, \cdots, N_1$,  then  $\Delta \mathscr{P}(t)=0$ and $\Delta \mathscr{D}(t)=0$ since all variables keep unchanged. 
	\end{enumerate}
	The above two cases conclude that 
	\begin{align*}
	\Delta \mathscr{P}(t) \leq \left(1+\frac{1}{\theta} \right)\Delta \mathscr{D}(t), 
	\end{align*}
	for all $t$.
	
	%Let $P$ and $D$ be the primal objective value, respectively, computed by Alg.~\ref{pda1}. Then, $P=\sum_{t=1}^{\infty} \Delta \mathscr{P}(t)$ and $D=\sum_{t=1}^{\infty} \Delta \mathscr{D}(t)$; therefore, the result follows since
	%\begin{align*}
	%P \left(1+\frac{1}{\theta} \right) D \leq \left(1+\frac{1}{\theta} \right) OPT(\mathbf{A}),
	%\end{align*}
	%where the last inequality is due to the weak duality \cite{.}.

	\section{Proof of Theorem~\ref{theorem:expected-competitive ratio1}}\label{appendix:theorem:expected-competitive ratio1}
	First,  we  compare the expected  cost of transmitting uncoded packets  by Alg.~\ref{online-alg1} with the term $C \cdot \widetilde{x}(\infty)$ of the primal objective value in Eq.~(\ref{primal1:objective}) computed by Alg.~\ref{pda1}. Note that, for a given $u$, there must exist $\lfloor \Delta \widetilde{x}(t) \rfloor$ $k$'s such that $u+k \in [\widetilde{x}_{\text{pre}}(t), \widetilde{x}(t))$, i.e., Alg.~\ref{online-alg1} transmits $\lfloor \Delta \widetilde{x}(t) \rfloor$ uncoded packets in slot $t$; in addition, according to \cite{tseng2019online}, Alg.~\ref{online-alg1}  transmits one more uncoded packet with probability  $\Delta \widetilde{x}(t)-\lfloor \Delta \widetilde{x}(t) \rfloor$. 
	Thus, the expected number of uncoded packets transmitted by Alg.~\ref{online-alg1} in slot~$t$ is $\Delta \widetilde{x}(t)$; moreover,  the expected total number of uncoded packets transmitted by Alg.~\ref{online-alg1} is $\sum_{t=1}^{\infty} \Delta \widetilde{x}(t)=\widetilde{x}(\infty)$. We can obtain that the expected  cost of transmitting uncoded packets by Alg.~\ref{online-alg1} is $C\cdot \widetilde{x}(\infty)$, which is exactly the  value of the first term of the primal objective value in Eq.~(\ref{primal1:objective})
	computed by Alg.~\ref{pda1}.
	
	Second, we compare the expected number of packets left at queue $Q_1$ at the end of slot $t$ under Alg.~\ref{online-alg1} with the term $\widetilde{z}(t)$  of the primal objective value in Eq.~(\ref{primal1:objective}) computed by Alg.~\ref{pda1}. Note that the expected number of packets left at queue~$Q_1$ at the end of slot~$t$ under Alg.~\ref{online-alg1} is
	\begin{align*}
	\max\{N_1-\underbrace{n_2(t)}_{\text{coded packets}}-\underbrace{\sum_{\tau=1}^t \Delta \widetilde{x}(\tau)}_{\text{uncoded packets}},0\} = \max\{N_1 -n_2(t) - \widetilde{x}(t),0\}.
	\end{align*}
	\begin{enumerate}
		\item  If $\widehat{x}_i(t) <1$ for $i=n_2(t)+1, \cdots, N_1$, then the expected number of packets left at queue $Q_1$ at the end of slot~$t$ under Alg.~\ref{online-alg1} is less than the term $\widetilde{z}(t)$  of the primal objective value in Eq.~(\ref{primal1:objective}) computed by Alg.~\ref{pda1} because $N_1 -n_2(t) - \widetilde{x}(t) \leq \widetilde{z}(t)$ (by the primal feasibility in Eq.~(\ref{primal1:constraint-1}) of Alg.~\ref{pda1}). 
		\item If $\widehat{x}_i(t) \geq 1$ for $i=n_2(t)+1, \cdots, N_1$,  then both (the expected number of packets left at queue $Q_1$ at the end of slot~$t$ under Alg.~\ref{online-alg1} and the term $\widetilde{z}(t)$  of the primal objective value in Eq.~(\ref{primal1:objective}) computed by Alg.~\ref{pda1}) are zeros because $N_1 -n_2(t) - \widetilde{x}(t) \leq N_1 -n_2(t) - \sum_{i=n_2(t)+1}^{N_1} \widehat{x}_i(t) \leq 0$.
	\end{enumerate} 
	
	We  conclude that  the expected cost in Problem~\ref{problem1} incurred by Alg.~\ref{online-alg1} is less than or equal to the primal objective value in Eq.~(\ref{primal1:objective}) computed by Alg.~\ref{pda1}. Then, the result immediately follows from Theorem~\ref{theroem:competitive-ratio1}.

	\section{Proof of Theorem~\ref{theorem:optimal i}} \label{appendix:theorem:optimal i}
	The proof of the theorem needs the following technical lemma. 
	\begin{proposition}  \label{lemma:optimal-value}
		Given some $\alpha_i$ and $\beta_i$ such that $\alpha_i \leq \beta_i$ for all $i=1, \cdots, N_1$, the optimal objective value of the   linear program
		\begin{eqnarray*}  
			&\min& C \cdot \sum^{N_1}_{i=1} x_i+  \sum^{\infty}_{t=1} \sum^{N_1}_{i=1} z_{i}(t)\\
			&\text{s.t.} & x_i+z_i(t) \geq 1 \text{\,\,\,for all $i$ and\,\,\,} \alpha_i \leq t \leq  \beta_i;\\
			&& x_i, z_i(t) \geq 0 \text{\,\,\, for all $i$ and $t$}
		\end{eqnarray*}
		is $\sum^{N_1}_{i=1} \min\{\beta_i -\alpha_i+1, C\}$.
	\end{proposition} 
	
	\begin{proof}
		We compute the minimum objective value of the linear program via its dual program: 
		\begin{eqnarray*} \label{eq:dual-given-I}
			&\max&\sum^{N_1}_{i=1} \sum^{\beta_i}_{t=\alpha_i} w_i(t) \\
			&\text{s.t.}& \sum^{\beta_i}_{t=\alpha_i} w_i(t) \leq C \text{\,\,\,for all $i$}; \nonumber\\
			&& 0 \leq w_i(t) \leq 1  \text{\,\,\,for all\,\,} i \text{\,\,and\,\,} t. \nonumber
		\end{eqnarray*}
		Since $\sum^{\beta_i}_{t=\alpha_i} w_i(t) \leq C$ and $0 \leq w_i(t) \leq 1$, we can obtain $\sum^{\beta_i}_{t=\alpha_i} w_i(t) \leq \min\{\beta_i-\alpha_i+1, C\}$; thus, the dual objective value is  bounded above by
		\begin{eqnarray*}
			\sum^{N_1}_{i=1} \sum^{\beta_i}_{t=\alpha_i} w_i(t) \leq \sum^{N_1}_{i=1}  \min\{\beta_i-\alpha_i+1, C\}.
		\end{eqnarray*}
		%chumiboda
		The equality in the above equation is achievable by setting $w_i(t)$ to be one for all $\alpha_i \leq t \leq \min\{\beta_i, \alpha_i+C-1\}$. Therefore, according to the duality theory, we can conclude that the optimal value of the linear program is 
		%chumibudada 
		$\sum^{N_1}_{i=1} \min\{\beta_i - \alpha_i+1, C\}$.
	\end{proof}
	
	Next, we prove the theorem by induction on $N_1$ and $N_2$. First, when $N_1=N_2=1$,    linear program~(\ref{primal2}) obviously can solve Problem~\ref{problem2}. Suppose that, when $N_1=n_1$ and $N_2=n_2$,   linear program~(\ref{primal2})  can solve Problem~\ref{problem2}.  Next, we  show that linear program~(\ref{primal2}) can also solve Problem~\ref{problem2} when $(N_1, N_2)$ is $(n_1+1, n_2)$, $(n_1, n_2+1)$, or $(n_1+1, n_2+1)$. We will focus on the case of $N_1=n_1+1, N_2=n_2+1$, while the other cases just follow the same arguments.

	By $\mathbf{A} - \{A_1(t_1),A_2(t_2)\}$ we denote the  arrival pattern  obtained by  removing a packet arriving at queue~$Q_1$ in slot $t_1$ and a packet arriving at queue~$Q_2$ in slot $t_2$ from  arrival pattern $\mathbf{A}$. By $\text{OPT}_{(\ref{problem2})}(\mathbf{A})$ we define the minimum  cost in Problem~\ref{problem2} under arrival pattern $\mathbf{A}$. Let $T^{(i)}_{j}$ be the slot when the $j$-th packet arrives at queue $Q_i$.  Let $i^*=\max\{i: T^{(1)}_{i} \leq T^{(2)}_1\}$ indicate a packet  arriving at queue~$Q_1$ in the slot closest to $T^{(2)}_1$.  Then, we can express $\text{OPT}_{(\ref{problem2})}(\mathbf{A})$ as
	\begin{align}
		\text{OPT}_{(\ref{problem2})}(\mathbf{A})= \min_{1\leq i \leq i^*} \bigl\{\min\{T^{(2)}_1-T^{(1)}_{i}+1,C\}+\text{OPT}_{(\ref{problem2})}(\mathbf{A}-\{A_1(T^{(1)}_{i}),A_2(T^{(2)}_1)\})\bigr\},  \label{eq:opt}
	\end{align}
	where the first term $\min\{T^{(2)}_1-T^{(1)}_{i}+1,C\}$ manages the $i$-th packet arriving  at queue~$Q_1$ (by slot~$T^{(2)}_1$)  and the first packet at queue~$Q_2$: if $T^{(2)}_1-T^{(1)}_{i}+1 \leq C$, then the $i$-th packet at queue~$Q_1$ optimally waits for coding with the first packet at queue~$Q_2$; otherwise, both packets  are optimally transmitted immediately (without coding) upon arrival. The second term $\text{OPT}_{(\ref{problem2})}(\mathbf{A}-\{A_1(T^{(1)}_{i}),A_2(T^{(2)}_1)\})$ expresses the minimum cost in Problem~\ref{problem2}  when the $i$-th packet at queue~$Q_1$ and the first packet at queue~$Q_2$ are both removed from arrival pattern $\mathbf{A}$.  
	
	%Moreover, according to Lemma~\ref{lemma:optimal-value} the minimum objective value in Eq.~(\ref{integer2:objective}) subject to set $\mathbf{I}$ from Alg.~\ref{i-alg} is $\min\{T^{(2)}_1-T_{1,i^*}+1,C\}+\text{OPT}(\mathbf{A}-\{A_1(T^{(1)}_{i}),A_2(T^{(2)}_1)\})$.
	%Thus, we aim to  show
	%\begin{align*}
	%&\min_{1\leq i \leq i^*} \bigl\{\min\{T^{(2)}_1-T^{(1)}_{i}+1,C\}\\
	%&+\text{OPT}(\mathbf{A}-\{A_1(T^{(1)}_{i}),A_2(T^{(2)}_1)\})\bigr\}\\
	%=&\min\{T^{(2)}_1-T_{1,i^*}+1,C\}+\text{OPT}(A-\{T_{1,i^*},T^{(2)}_1\}),
	%\end{align*}
	%for all possible arrival patterns $\mathbf{A}$.

	By the induction hypothesis,   linear program~(\ref{primal2}) can solve $\text{OPT}_{(\ref{problem2})}(\mathbf{A}-\{A_1(T^{(1)}_{i}),A_2(T^{(2)}_1)\})$. Particularly, for  arrival pattern $\mathbf{A}-\{A_1(T^{(1)}_{1}),A_2(T^{(2)}_1)\}$, we assume that Alg. \ref{i-alg} produces  constraints  $x_i+z_i(t) \geq 1$ for  $T^{(1)}_i \leq t \leq \beta^*_i$ (for some $T^{(1)}_i \leq \beta^*_i <\infty$\footnote{If $\beta^*_i$ generated by Alg.~\ref{i-alg} is infinity, then we can arbitrarily choose a large number as $\beta^*_i$ according to Proposition~\ref{lemma:optimal-value}}) and $i=2, \cdots, i^*$. By Proposition~\ref{lemma:optimal-value}, we can express $\text{OPT}_{(\ref{problem2})}(\mathbf{A}-\{A_1(T^{(1)}_{1}),A_2(T^{(2)}_1)\})$ in Eq.~(\ref{eq:opt}) as
	\begin{align}
		\text{OPT}_{(\ref{problem2})}(\mathbf{A}-\{A_1(T^{(1)}_{1}),A_2(T^{(2)}_1)\}) =\sum_{i=2}^{i^*} \min\{\beta^*_i - T^{(1)}_i+1, C\} +R, \label{eq:opt1}
	\end{align}
	where $R$ is  the remaining cost  incurred by the packets  arriving at queue~$Q_1$ after slot $T^{(2)}_1$. Similarly,  we can express $\text{OPT}_{(\ref{problem2})}(\mathbf{A}-\{A_1(T^{(1)}_{i}),A_2(T^{(2)}_1)\})$ in Eq.~(\ref{eq:opt}) by
	\begin{align}
		&\text{OPT}_{(\ref{problem2})}(\mathbf{A}-\{A_1(T^{(1)}_i),A_2(T^{(2)}_1)\})\nonumber\\
		=& \sum_{j=1}^{i-1} \min\{\beta^*_{j+1} - T^{(1)}_j+1, C\} + \sum_{j=i}^{i^*-1} \min\{\beta^*_{j+1} - T^{(1)}_{j+1}+1, C\}+R, \label{eq:opt2}
	\end{align}
	for  $2\leq i\leq i^*-1$, and 
	\begin{align}
		\text{OPT}_{(\ref{problem2})}(\mathbf{A}-\{A_1(T^{(1)}_{i^*}),A_2(T^{(2)}_1)\})=\sum_{j=1}^{i^*-1} \min\{\beta^*_{j+1} - T^{(1)}_j+1, C\}+R.  \label{eq:opt2'}
	\end{align} 
	
	By Proposition~\ref{lemma:optimal-value} again,  the minimum objective value in linear program~(\ref{primal2}) under arrival pattern~$\mathbf{A}$ can be expressed by
	\begin{align}
		\min\{T^{(2)}_1-T^{(1)}_{i^*}+1,C\}+\sum_{j=1}^{i^*-1} \min\{\beta^*_{j+1} - T^{(1)}_j+1, C\}+R.  \label{eq:opt3}
	\end{align} 
	From Eqs.~(\ref{eq:opt}) - (\ref{eq:opt3}), it suffices to  show that
	\begin{align*}
		& \min\{T^{(2)}_1-T^{(1)}_{i^*}+1,C\}+\sum_{j=1}^{i^*-1} \min\{\beta^*_{j+1} - T^{(1)}_j+1, C\}\\
		\leq& \min\{T^{(2)}_1-T^{(1)}_i+1,C\}+\sum_{j=1}^{i-1} \min\{\beta^*_{j+1} - T^{(1)}_j+1, C\} + \sum_{j=i}^{i^*-1} \min\{\beta^*_{j+1} - T^{(1)}_{j+1}+1, C\},
	\end{align*}
	for all $1 \leq i \leq i^*-1$. By removing the common terms from  both sides of the above equation, it suffices to show that 
	\begin{align}
		& \min\{T^{(2)}_1-T^{(1)}_{i^*}+1,C\}+\sum_{j=i}^{i^*-1} \min\{\beta^*_{j+1} - T^{(1)}_j+1, C\}\nonumber\\
		\leq& \min\{T^{(2)}_1-T^{(1)}_i+1,C\}+ \sum_{j=i}^{i^*-1} \min\{\beta^*_{j+1} - T^{(1)}_{j+1}+1, C\}, \label{eq:opt-condition}
	\end{align}
	for all $1 \leq i \leq i^*-1$.  
	
	\begin{figure}
		\centering
		\includegraphics[width=.6\textwidth]{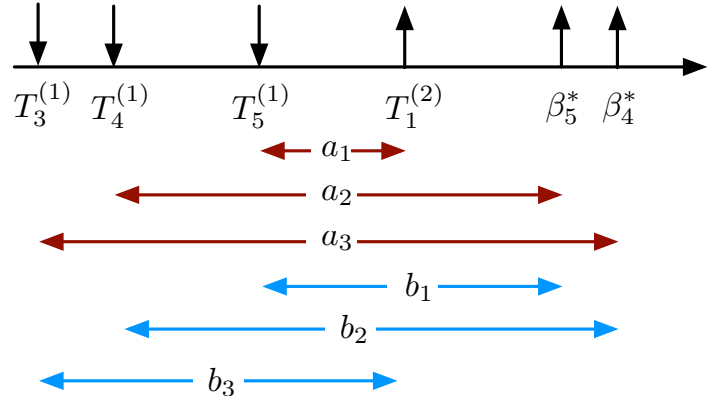}
		\caption{Illustration of notations $a_j$ and $b_j$ when $i=3$ and $i^*=5$ are given.}
		\label{fig:proof}
	\end{figure}

	For brevity, for a fixed $i$ in Eq.~(\ref{eq:opt-condition}), we denote $a_j$  by $a_1=T^{(2)}_1-T^{(1)}_{i^*}+1$ and $a_j=\beta^*_{i^*-j+2}-T^{(1)}_{i^*-j+1}+1$ for $2 \leq j \leq i^*-i+1$;  denote $b_j$  by $b_j=\beta^*_{i^*-j+1}-T^{(1)}_{i^*-j+1}+1$ for $1\leq j \leq i^*-i$ and  $b_{i^*-i+1}=T^{(2)}_1-T^{(1)}_i+1$. See Fig.~\ref{fig:proof} for illustrating notations $a_j$ and $b_j$.  Let $i_{\max}=i^*-i+1$. With the set of notations, Eq.~(\ref{eq:opt-condition}) can be simplified as 
	\begin{align}
		\sum_{j=1}^{i_{\max}} \min\{a_j,C\} \leq \sum_{j=1}^{i_{\max}} \min\{b_j,C\}.  \label{eq:opt-condition2}
	\end{align}
	To verify Eq.~(\ref{eq:opt-condition2}),  note (see Fig.~\ref{fig:proof} for example) that
	\begin{align}
		&a_1 \leq b_1 \leq  a_2 \leq b_2 \cdots  \leq a_{i_{\max}};\label{eq:fact1}\\
		&a_1+\cdots+a_{i_{\max}}=b_1+\cdots+b_{i_{\max}}.\label{eq:fact2}
	\end{align}
	Then,  Eq. (\ref{eq:opt-condition2}) can be confirmed by the following cases. 
	\begin{enumerate}
		\item \textit{If $b_{i_{\max}}  \geq C$}:
		\begin{itemize}
			\item \textit{If $a_{1} \geq C$}, then 
			\begin{align*}
				\sum_{j=1}^{i_{\max}} \min\{a_j, C\}={i_{\max}}\cdot C=\sum_{j=1}^{i_{\max}} \min\{b_j, C\}.
			\end{align*}
			\item \textit{If $a_k \leq C \leq b_k$ for some $k=1, \cdots, {i_{\max}}-1$}, then 
			\begin{align*}
				\sum_{j=1}^{i_{\max}} \min\{a_j, C\} = & \sum^{k-1}_{j=1} a_j + a_k+ ({i_{\max}}-k)C \\
				\mathop{\leq}^{(a)} &\sum_{j=1}^{k-1} b_j+({i_{\max}}-k+1)C\\
				=&\sum^{i_{\max}}_{j=1} \min\{b_j, C\},
			\end{align*}
			where (a) is from Eq.~(\ref{eq:fact1}) and $a_k \leq C$.
			\item \textit{If $b_{k} \leq C \leq a_{k+1}$ for some $k=1, \cdots, {i_{\max}}-1$}, then  
			\begin{align*}
				\sum_{j=1}^{i_{\max}} \min\{a_j, C\} =  \sum^{k}_{j=1} a_j +  ({i_{\max}}-k)C \mathop{\leq}^{(a)} \sum_{j=1}^{k} b_j+({i_{\max}}-k)C=\sum^{i_{\max}}_{j=1} \min\{b_j, C\},
			\end{align*}
			where (a) is due to Eq.~(\ref{eq:fact1}).
			\item \textit{If $a_{{i_{\max}}} \leq C$}, then  
			\begin{align*}
				\sum_{j=1}^{i_{\max}} \min\{a_j, C\} =  \sum^{{i_{\max}}-1}_{j=1} a_j+a_{i_{\max}} \mathop{\leq}^{(a)} \sum_{j=1}^{{i_{\max}}-1} b_j+C=\sum^{i_{\max}}_{j=1} \min\{b_j, C\},
			\end{align*}  
			where (a) is from to Eq.~(\ref{eq:fact1}) and $a_{i_{\max}} \leq C$.
			
		\end{itemize}			
		\item \textit{If $b_{i_{\max}} < C$}:
		\begin{itemize}
			\item The case of $a_1 \geq C$ is impossible because $a_1 \leq b_{i_{\max}}< C$. 
			
			\item \textit{If $a_k \leq C \leq b_k$ for some $k=1, \cdots, {i_{\max}}-1$}, then 
			\begin{align*}
				\sum^{i_{\max}}_{j=1} \min\{a_j, C\}=&\sum_{j=1}^k a_j+({i_{\max}}-k)C\\
				\mathop{\leq}^{(a)}& \sum_{j=1}^k a_j+({i_{\max}}-k)C+ \left(\sum^{{i_{\max}}}_{j=k+1} a_j -\sum_{j=k}^{{i_{\max}}-1} b_j\right)\\
				\mathop{=}^{(b)}&\sum_{j=1}^{i_{\max}} b_j - \sum^{{i_{\max}}-1}_{j=k}b_j+ ({i_{\max}}-k)C\\
				=&\sum^{k-1}_{j=1}b_j+b_{i_{\max}}+({i_{\max}}-k) C\\
				=&\sum^{{i_{\max}}}_{j=1}\min\{b_j, C\}, 
			\end{align*} 			 
			where   (a) is from Eq.~(\ref{eq:fact1}) and  (b) is from Eq.~(\ref{eq:fact2}). 
			\item \textit{If $b_{k} \leq C \leq a_{k+1}$ for some $k=1, \cdots, {i_{\max}}-1$}, then
			\begin{align*}
				\sum^{i_{\max}}_{j=1} \min\{a_j, C\}=&\sum_{j=1}^k a_j+({i_{\max}}-k)C\\
				\mathop{\leq}^{(a)}& \sum_{j=1}^{k+1} a_j+({i_{\max}}-k-1)C\\
				\mathop{\leq}^{(b)}& \sum_{j=1}^{k+1} a_j+({i_{\max}}-k-1)C+ \left(\sum^{i_{\max}}_{j=k+2} a_j -\sum_{j=k+1}^{{i_{\max}}-1} b_j\right)\\
				\mathop{=}^{(c)}&\sum_{j=1}^{i_{\max}} b_j - \sum^{{i_{\max}}-1}_{j=k+1}b_j+ ({i_{\max}}-k-1)C\\
				=&\sum^{k}_{j=1}b_j+b_{i_{\max}}+({i_{\max}}-k-1) C\\
				=&\sum^{{i_{\max}}}_{j=1}\min\{b_j, C\}, 
			\end{align*}
			where (a) is from $a_{k+1} \geq C$; (b) is from Eq.~(\ref{eq:fact1}); (c) is from Eq.~(\ref{eq:fact2}). 
			\item    \textit{If $a_{i_{\max}} \leq C$}, then
			\begin{align*}
				\sum_{j=1}^{i_{\max}} \min\{a_j, C\}=\sum_{j=1}^{i_{\max}} a_j \mathop{=}^{(a)} \sum_{j=1}^{i_{\max}} b_j= \sum_{j=1}^{i_{\max}} \min\{b_j, C\},
			\end{align*}
			where (a) is from Eq.~(\ref{eq:fact2}).
		\end{itemize}
		
	\end{enumerate}
	Then, we complete the proof.

	\section{Proof of Theorem~\ref{theorem:comp-ratio-two-side}} \label{appendix:theorem:comp-ratio-two-side}
	We follow the  notations  in the  proof of Theorem \ref{theroem:competitive-ratio1}.  For each slot $t$,  the change of the primal objective value in Eq.~(\ref{primal2:objective}) (under Alg.~\ref{pda2}) is
	\begin{align*}
	\Delta \mathscr{P}(t) &= \sum_{i=1}^{N_1} C \cdot (\widetilde{x}_i(t)-\widehat{x}_i(t)) + \widetilde{z}_i(t) \\
	&=  \sum_{i \in \mathbf{I}(t)- \{i: \widehat{x}_i(t) \geq 1\}} C \cdot (\frac{\widehat{x}_i(t)}{C}+\frac{1}{\theta \cdot C}) + (1- \widehat{x}_i(t))\\
	&=|\mathbf{I}(t)- \{i: \widehat{x}_i(t) \geq 1\}| ( 1 + \frac{1}{\theta}).
	\end{align*}
	Moreover,   the change of the dual objective value in Eq.~(\ref{dual2:objective}) (under Alg.~\ref{pda2}) is $\Delta \mathscr{D}(t)=|\mathbf{I}(t)- \{i: \widehat{x}_i(t) \geq 1\}|\widetilde{w}(t) = |\mathbf{I}(t)- \{i:\widehat{x}_i(t) \geq 1\}|$. Then, following the  line in the proof of Theorem~\ref{theroem:competitive-ratio1} yields the result.

	\begin{small}
		\bibliographystyle{IEEEtran}
		\bibliography{IEEEabrv,ref}
	\end{small}

	%\begin{IEEEbiography}
	%	[{\includegraphics[width=1in,height=1.25in,clip,keepaspectratio]{yupin.pdf}}]{Yu-Pin Hsu} received his BS and MS degrees from National Chiao Tung University, Taiwan, in 2005 and 2007 respectively, and his PhD degree from Texas A\&M University, USA, in 2014. He was post-doctoral fellows with Singapore University of Technology and Design, Singapore, in 2014, and with Massachusetts Institute of Technology, USA, in 2015. In 2016, he joined National Taipei University, Taiwan, where he is currently an assistant professor. His research interests center around communication networks with focus on algorithmic and control-theoretic aspects.
	%\end{IEEEbiography}

\end{document}